\documentclass[a4paper,onecolumn,11pt,unpublished]{quantumarticle}
\pdfoutput=1
\usepackage[utf8]{inputenc}
\usepackage[english]{babel}
\usepackage[T1]{fontenc}
\usepackage{amsmath}
\usepackage{hyperref}
\usepackage[numbers]{natbib}
\usepackage{tikz}
\usepackage{lipsum}
\usepackage{subfigure}
\usepackage{bm}
\usepackage{amssymb}
\usepackage{amsthm}
\usepackage{multirow}
\usepackage{multicol}
\usepackage[ruled]{algorithm2e}

\newcommand{\astate}{{\mathbf{s}}}
\newcommand{\anode}{{\mathbf{v}}}
\newcommand{\apath}{{\mathbf{p}}}
\newcommand{\name}{GSA}

\newtheorem{theorem}{Theorem}

\newtheorem{lemma}{Lemma}
\newtheorem{problem}{Problem}
\newtheorem{example}{Example}

\begin{document}

\title{An Efficient Gradient-Sensitive Alternate Framework for VQE with Variable Ansatz}

\author[1,3,4]{Ze-Tong Li}
\author[2,3,4]{Fan-Xu Meng}
\thanks{Equal contribution to Ze-Tong Li}
\author[2,3,4]{Han Zeng}
\author[2,3,4,5]{Zai-Chen Zhang}
\author[1,3,4,5]{Xu-Tao Yu}
\email{yuxutao@seu.edu.cn}

\affil[1]{State Key Laboratory of Millimeter Waves, Southeast University, Nanjing 210096, China}
\affil[2]{National Mobile Communications Research Laboratory, Southeast University, Nanjing 210096, China.}
\affil[3]{Frontiers Science Center for Mobile Information Communication and Security, Southeast University, Nanjing 210096, China.}
\affil[4]{Quantum Information Center, Southeast University, Nanjing 210096, China.}
\affil[5]{Purple Mountain Lab, Nanjing 211111, China.}

\maketitle

\begin{abstract}
  Variational quantum eigensolver (VQE), aiming at determining the ground state energy of a quantum system described by a Hamiltonian on noisy intermediate scale quantum (NISQ) devices, is among the most significant applications of variational quantum algorithms (VQAs). However, the accuracy and trainability of the current VQE algorithm are significantly influenced due to the \emph{barren plateau} (BP), the non-negligible gate error and limited coherence time in NISQ devices. To tackle these issues, a gradient-sensitive alternate framework with variable ansatz is proposed in this paper to enhance the performance of the VQE. We first propose a theoretical framework for VA-VQE via alternately solving a multi-objective optimization problem and the original VQE, where the multi-objective optimization problem is defined with respect to cost function values and gradient magnitudes. Then, we propose a novel implementation method based on the double $\epsilon$-greedy strategy with the candidate tree and modified multi-objective genetic algorithm. As a result, the local optima are avoided both in ansatz and parameter perspectives, and the stability of output ansatz is enhanced. The experimental results indicate that our framework shows considerably improvement of the error of the found solution by up to 87.9\% compared with the hardware-efficient ansatz. Furthermore, compared with the full-randomized VA-VQE implementation, our framework is able to obtain the improvement of the error and the stability by up to 36.0\% and 58.7\%, respectively, with similar quantum costs.
\end{abstract}

\section{Introduction}
Quantum computing, a promising paradigm for solving many classically intractable problems \cite{grover1997Quantum,shor1999Polynomialtime,harrow2009Quantum,low2019Hamiltonian}, is confronted with the era of noisy intermediate scale quantum
(NISQ) \cite{preskill2018Quantum}. NISQ devices, characterized as tiny quantities and short coherent time of qubits and low fidelity of quantum gates, seem difficult to efficiently apply the extraordinary quantum algorithms. 

A hybrid quantum-classical scheme kindles the hope of extracting the quantum advantages from NISQ devices. In this scheme, the quantum computer maintains a parameterized quantum circuit (aka. ansatz), and the classical computer runs an optimizer to find the parameters that minimize the cost function of the ansatz. This novel scheme was first proposed as the variational quantum eigensolver (VQE) \cite{peruzzo2014Variational} to find the ground-state energy, and generalized for solving linear systems \cite{bravo-prieto2020Variational}, simulating dynamics \cite{yuan2019Theory}, decomposing matrix \cite{wang2021Variational,meng2021Quantum}, reducing the dimensionality of data \cite{li2022Quantum}, and solving problems in the data science domain \cite{beer2020Training,abbas2021Power,biamonte2017Quantum}. In the following text, the term VQE is used to indicate all methods that adopt the aforementioned scheme.

The VQE has experimentally shown effectiveness in small-scale problems \cite{havlicek2019Supervised}. The elaborately engineered ansatzes \cite{kandala2017Hardwareefficient} with fixed layer structure and adjustable cardinality of layers have impressively exhibited high expressibility, the ability to express the extensive range of quantum states. However, the performance of VQE with the ansatzes degrades significantly with the qubit number and circuit depth \cite{du2020Quantum}. The phenomena of the degradation are generally exhibited as the non-ignorable absolute error between the converged cost function value and the exact minimum. 

One reason for the degradation is that the VQE suffers from the so-called barren plateau \cite{mcclean2018Barren} that the gradient magnitudes vanish exponentially with the system scale. Then, the severity of the BP phenomenon is linked to the number of entanglements \cite{patti2021Entanglement, ortizmarrero2021EntanglementInduced,cerezo2021Cost,sack2022Avoiding} and the expressibility \cite{holmes2022Connecting}. Another reason is quantum hardware noise both in quantum gates and qubits. The influence of the noise accumulates with the number of quantum gates and the circuit depth \cite{wang2021Noiseinduced} and causes the accuracy reduction of the observable estimation. This significantly impedes the convergence of the cost function since the value of the cost function and its gradient are inaccurate, which is known as noise-induced barren plateaus \cite{wang2021Noiseinduced}.

Recently, the variable ansatz strategy \cite{du2020Quantum,rattew2020Domainagnostic,zhang2021Differentiable,ostaszewski2021Reinforcement,bilkis2022Semiagnostic,zhang2021Mutual,meng2021Quantuma} has emerged as a promising technique to reserve the quantum advantage in large-scale situations. Unlike the structure fixed ansatz, the layer constraint is relaxed in variable ansatz, and gates can be added anywhere in the ansatz as required. As a result, both the size and the depth of the constructed anzatz can be effectively reduced, and thus better solutions can be found. Nevertheless, implementing the VQE with variable ansatz (VA-VQE) exists a significant space overhead for the ansatz layout optimization procedure, and consequently requires substantially more computing resources than that with the structure fixed ansatz.

To further improve efficiency, several techniques from intelligent algorithms have also been introduced. Several methods based on the complete training are proposed \cite{rattew2020Domainagnostic,ostaszewski2021Reinforcement,bilkis2022Semiagnostic,chivilikhin2020MoGVQE}. In these methods, the performance estimation of an ansatz requires VQE training until the cost function converges to obtain the lowest cost function value of the ansatz. As a result, high precision of the optimal selection among sampled ansatzes is achieved but undesirably introduces tremendous cost function determinations which are time-consuming because each determination of cost function requires substantial quantum measurements. Therefore, the range of sampled ansatzes is narrow when the quantity of quantum measurements is limited. This may result in local optimum from the ansatz perspective.

Attempting to overcome this hurdle, several methods based on the highly shared parameter pool are proposed \cite{du2020Quantum,meng2021Quantuma}. These methods construct ansatzes layer-wisely and maintain parameter pools from which values of parameters are exploited to calculate the cost function value as the performance predictor. Introducing the weight-sharing policy \cite{elsken2019Neural}, parameters are shared among layers with similar structures. Thus, the number of parameters in the parameter pool required to be trained is effectively reduced. These methods first train the parameter pools iteratively. In each iteration, several ansatzes are sampled based on the sampling strategy. Then, the parameters corresponding to the ansatzes are updated in specified steps (especially one step) by an optimizer. After training the parameter pools, the optimal ansatz selection procedure is processed. Quantities of ansatzes are sampled and evaluated by the performance predictor. Finally, the ansatz with the best performance, generally the lowest cost function value, is selected as the optimal ansatz. 

This kind of method gains the efficiency in finding optimal ansatz by avoiding complete training of ansatzes. However, since the parameters are highly shared and the training step of each ansatz is insufficient to converge, the current cost function values may not express the actual performances of the ansatzes, and hence the output suffers from high variance. This may result from two issues. The first is the training competitions that values of parameters in `good' ansatzes with promising layouts may be disturbed by the training of `bad' ansatzes incurred by the weight-sharing policy. The second issue is that the policy for ansatz layout optimization exploits the cost function value only as the benchmark to find the optimal ansatz structure. The mediocre ansatz whose cost function value is currently the lowest and cannot be further decreased may be selected as the optimum. Contrarily, the promising ansatzes may be discarded because of their temporarily high cost function values.

In this paper, addressing these issues, we propose a gradient-sensitive alternate framework (\name) for VA-VQE. The \name~alternately optimizes the structure of the ansatz and corresponding parameters by emphasizing the gradient magnitudes. Contributions are listed below:
\begin{itemize}
  \item We first propose a theoretical framework solving VA-VQE via alternately for solving a gradient magnitudes related multi-objective optimization problem and the original VQE, so that local optima can be avoided from the ansatz perspective and the stability is enhanced compared to other VA-VQE methods.
  \item To mitigate training competitions in training the parameter pool for the initialization of parameters, we first exploit the double $\epsilon$-greedy strategy based on the candidate tree to differentiate ansatzes by their cost function values and the gradient magnitudes, so that local optima can be evaded from the parameter perspective.
  \item We further reduce the size of the search space of the ansatz via applying gate commutation rules and establishing a bijection between the search space and the practical implementations of ansatzes to boost the time efficiency of the optimal ansatz and parameter determination.
  \item We adopt relatively fair criteria for measuring the performance of VA-VQE so that the transverse comparison among methods of VA-VQE can be conducted. As a result, the \name~shows conspicuously better performance on average compared to the structure fixed HEA up to 87.9\% improvement in terms of absolute error, and to the full-randomized VA-VQE method up to 36.0\% and 58.7\% improvement in terms of absolute error and stability (mean square error), respectively, with similar quantum costs (the number of calculations of the cost function).
\end{itemize}


This paper is structured as follows: In Sec.~\ref{sec:background}, we give a brief introduction to the basic knowledge. Then, the gradient-sensitive theoretical framework based on the alternate optimization is proposed in Sec.~\ref{sec:theoretical_framework}. The detail of the practical implementation of \name~is proposed in Sec.~\ref{sec:practical_implementation}. Subsequently, we conduct numerical simulations with relatively fair criteria in Sec.~\ref{sec:simulations} to show the advantages of our proposed method. Finally, we conclude this work in Sec.~\ref{sec:conclusion}. Notice that the examples and pseudocodes of algorithms are summarized in Appendix \ref{appendix:example} and \ref{appendix:alg}, respectively.

\section{Background}\label{sec:background}

\subsection{Variational Quantum Eigensolver}
In this work, we address the variational quantum eigensolver (VQE) tasks identified by sets of tuples $\mathbb{T} = \left\{\left(O_i, \rho_i, f_i\right)\right\}$ aiming at minimizing a cost function
\begin{equation}\label{eq:cost}
  C(\bm{\theta})=\sum_{i} f_i({\rm Tr}\left[O_i U(\bm{\theta}) \rho_i U^\dagger(\bm{\theta})\right]),
\end{equation}
where $\{\rho_i\}$ is the training set formed as n-qubit quantum states, $U(\bm{\theta})$ is a specified parameterized quantum circuit (aka. ansatz) with parameters $\bm{\theta}$, $O_i$ are observables and $f_i$ are bounded second-order differentiable functions that encode the problem with respect to the operand observable $O_i$ and state $\rho_{i}$. Generally,  while implementing the VQE, the quantum computer applies the ansatz $U(\bm{\theta})$ and processes the quantum measurements to determine expectations ${\rm Tr}\left[O_i U(\bm{\theta}) \rho_i U^\dagger(\bm{\theta})\right],\forall i$. The classical computer computes the cost function value and runs an optimization algorithm to find the optimal parameters $\bm{\theta}^*$ that minimize the cost function. We provide an example finding the ground state energy of a quantum system described by a Hamiltonian in Eg.~\ref{eg:vqe}.

\subsection{Quantum Gradient}
In the optimization procedure of large-scale VQE, the gradient-based methods (e.g., gradient descent) are generally more preferred than gradient-free methods (e.g., Nelder-Mead) \cite{mitarai2018Quantum}. In the gradient-based methods, the gradient of the cost function with respect to the parameters is essentially estimated in each optimization iteration.

Without loss of generality, an ansatz can be mathematically defined by
\begin{equation}\label{eq:ansatz}
  U(\bm{\theta}) := \prod_{l=1}^{N_l} U_l(\theta_l)W_l,
\end{equation}
where $U_l(\theta_l) = \exp (-i\theta_lV_l)$, $V_l$ is a Hermitian operator, and $W_l$ is a non-parametrized quantum gate. Then, the partial derivative of an expectation $E(\bm{\theta})={\rm Tr}(OU(\bm{\theta})\rho U^\dagger(\bm{\theta}))$ with respect to the $k$th parameter $\theta_k$ is
\begin{equation}
  \begin{aligned}
    \partial_k E(\bm{\theta}) &\equiv \frac{\partial E(\bm{\theta})}{\partial \theta_k}
    = i{\rm Tr}\left(\left[V_k,U_L^\dagger O U_L\right] U_R \rho U_R^\dagger \right),
  \end{aligned}
\end{equation}
where we use the notations
\begin{align}
  U_R \equiv \prod_{l=1}^{k-1} U_l(\theta_l)W_l,~U_L \equiv \prod_{l=k}^{N_l} U_l(\theta_l)W_l.
\end{align}

In this paper, we apply the parameter-shift rule \cite{mitarai2018Quantum} for gradient estimating. We assume that all parameterized quantum gates are Pauli rotations. Therefore, the partial derivative of $E(\bm{\theta})$ with respect to $\theta_k$ is
\begin{equation}
  \partial_k E(\bm{\theta}) = \frac{1}{2}\left[E(\bm{\theta}+\frac{\pi}{2}\bm{e}_k)-E(\bm{\theta}-\frac{\pi}{2}\bm{e}_k)\right],
\end{equation}
where $\bm{e}_k$ is a vector whose $k$th element is $1$ and others are $0$.

\subsection{Gradient Descent}\label{subsec:gradient_descent}
Gradient descent is literately a gradient-based optimizer that has been generally used in training VQE \cite{bravo-prieto2020Variational,yuan2019Theory,li2022Quantum,beer2020Training}. The kernel process of the gradient descent can be mathematically represented as 
\begin{equation}
  \bm{\theta} \leftarrow \bm{\theta} - \alpha \nabla C(\bm{\theta}),
\end{equation}
where $\alpha$ is the step size. 

In this paper, we exploit a line search to estimate the step size $\alpha$ instead of a fixed one. In each optimization step, $\alpha$ satisfies the Wolfe conditions \cite{wolfe1969Convergence} 
\begin{gather}
  C(\bm{\theta}-\alpha\bm{g})\le C(\bm{\theta}) - c_1\alpha\left\|\bm{g}\right\|_2^2,\label{eq:wolfe_cond_1}\\
  \nabla C(\bm{\theta}-\alpha\bm{g})^T\bm{g}\ge c_2 \left\|\bm{g}\right\|_2^2,
\end{gather}
where $\bm{g}=\nabla C(\bm{\theta})$ and $0<c_1<c_2<1$. To facilitate the $\alpha$ determination, given a reference step size $\alpha_0$, we gradually decrease $\alpha$ from $\alpha_0$ by repeating $\alpha\leftarrow\rho\alpha$ until Eq.~(\ref{eq:wolfe_cond_1}) establishes, where $\rho<0$. In this paper, we empirically set $c_1 = 10^{-4}$ and $\rho=0.618$.

\subsection{Expressibility}
In the absence of prior knowledge about the solution unitaries $\mathbb{U}_s$ of a VQE task, the ability of ansatz to generate a wide range of unitaries $\mathbb{U}$ to guarantee $\mathbb{U}_s\cap\mathbb{U}\ne \emptyset$ is required. The expressibility of an ansatz describes the degree to which it uniformly explores the unitary group $\mathcal{U}(2^n)$, and can be simply considered as the range of unitaries the ansatz can generate. By comparing the uniform distribution of unitaries obtained from $\mathbb{U}$ to the Haar distribution of unitaries from $\mathcal{U}$, the expressibility of an ansatz can be defined by the superoperator \cite{sim2019Expressibility,nakaji2021Expressibility}:
\begin{equation}
  \begin{aligned}
    \mathcal{A}_u^{(t)}:=&\int_{\mathcal{U}(2^n)} d \mu(V)V^{\otimes t}(\,\cdot\,)(V^\dagger)^{\otimes t}-\int_{\mathbb{U}}d U\,U^{\otimes t}(\,\cdot\,)(U^\dagger)^{\otimes t},
  \end{aligned}
\end{equation}
where $d \mu(V)$ is the volume element of the Haar measure and $dU$ is the volume element corresponding to the uniform distribution over $\mathbb{U}$. Here we are especially interested in the expressibility of the ansatz with respect to the input quantum state $\rho$ and the observable $O$
\begin{align}
  \epsilon_{\mathbb{U}}^\rho := \left\|\mathcal{A}^{(2)}_\mathbb{U}(\rho^{\otimes 2})\right\|_2,~\epsilon_{\mathbb{U}}^O := \left\|\mathcal{A}^{(2)}_\mathbb{U}(O^{\otimes 2})\right\|_2.
\end{align}
Small values of $\epsilon_{\mathbb{U}}^\rho$ and $\epsilon_{\mathbb{U}}^O$ indicate the high expressibility of the ansatz.

\subsection{Barren Plateau and Trainability}
As one of the key challenges of VQE, the barren plateau phenomenon exhibited the exponential decrement of the gradient magnitudes with respect to the system size $n$ \cite{mcclean2018Barren} and was generalized that the ansatzes' expressibility \cite{holmes2022Connecting} and the amount of entanglement \cite{patti2021Entanglement, ortizmarrero2021EntanglementInduced,cerezo2021Cost,sack2022Avoiding} play significant roles in leading to barren plateaus. 
We highlight the severity of the BP phenomenon with respect to $n$ and the expressibility of the ansatz by the limited variance of gradient magnitudes
\begin{equation}\label{eq:bp}
  \begin{aligned}
    {\rm Var}[\partial_k C(\bm{\theta})] \le \frac{g(\rho,O,U)}{2^{2n}-1}+f\left(\epsilon_{\mathbb{U}_L}^O,\epsilon_{\mathbb{U}_L}^\rho\right),
  \end{aligned}
\end{equation}
where $g(\rho,O,U)$ is the prefactor in $O(2^n)$,
\begin{equation}
  f(x,t)=4xy + \frac{2^{n+2}\left(x\left\|O\right\|^2_2 + y\left\|\rho\right\|^2_2\right)}{2^{2n}-1},
\end{equation}
and $\mathbb{U}_L$ and $\mathbb{U}_R$ are ensembles of $U_L$ and $U_R$, respectively \cite{holmes2022Connecting}. The first term on the right in Eq.~(\ref{eq:bp}) indicates the variance of 2-design ansatz and is in $O(1/2^n)$, and the second term is the expressibility-dependent correction. From the Chebyshev's inequality, the trainability of an ansatz can be described by 
\begin{equation}
  \begin{aligned}
    {\rm Pr}\left[\left|{\partial_k C(\bm{\theta})}\right|\ge \delta\right] \le\frac{{\rm Var}[\,\partial_k C(\bm{\theta})]}{\delta^2},\,\forall \delta>0.
  \end{aligned}
\end{equation}
When the ansatz exhibits the BP phenomenon, the probability decreases exponentially with respect to $n$, which indicates that the precision to determine a cost-minimizing direction is exponentially large to $n$ \cite{bilkis2022Semiagnostic,cerezo2021Higher,arrasmith2021Effect}. Notably, an ansatz with higher expressibility suffers lower trainability since it exhibits a more severe BP phenomenon.

Remarkably, there exists another kind of barren plateaus, the noise-induced barren plateaus \cite{wang2021Noiseinduced}, caused by the imperfect quantum hardware. The cost function value concentrates exponentially around its average as the influence of noise accumulates, since the noise models acting throughout the ansatz map the input state toward the fixed point of the noise model \cite{wang2021Noiseinduced,stilckfranca2021Limitations}. 

This challenging phenomenon cannot simply be escaped by changing the optimizer \cite{arrasmith2021Effect}. While several attempts have been made to mitigate the severity of the barren plateau \cite{pesah2021Absence,volkoff2021Large,skolik2021Layerwise,grant2019Initialization,verdon2019Learning}, it is widely accepted that the variable ansatz strategy is promising to address this issue via automatically balancing the expressibility, the influence of noise, and the trainability.

\subsection{Hardware Constraints}\label{subsec:hardware_constraints}
Generally, only a limited number of gates are available on a practical quantum computer and the two-qubit gates are only allowed to be applied on a specific set of qubit pairs. These available gates are known as native gates of the quantum hardware. In this paper, we assume the native gates on an $n$-qubit quantum system to be $R_y$, $R_z$ and CNOT mathematically represented as
\begin{gather}
  R^q_y(\theta) = e^{-i \frac{\theta}{2}\sigma_y} =\begin{bmatrix}
    \cos\frac{\theta}{2}& -\sin\frac{\theta}{2}\\
    \sin\frac{\theta}{2}&  \cos\frac{\theta}{2}
  \end{bmatrix},\\
  R^q_z(\theta) = e^{-i \frac{\theta}{2}\sigma_z} =\begin{bmatrix}
    e^{-i{\theta}/{2}}& 0\\
    0& e^{i{\theta}/{2}}
  \end{bmatrix},\\
  \mathrm{CNOT}^{q_a,q_b}=\left|0\right>\left<0\right|^{q_a}\otimes I^{q_b} + \left|1\right>\left<1\right|^{q_a}\otimes X^{q_b},
\end{gather}
where superscript $q$, $q_a$, and $q_b$ indicate the qubits on which the quantum gates act, $\sigma_y$ and $\sigma_z$ are Pauli operators mathematically represented as 
\begin{align}
  \sigma_y = \begin{bmatrix}
    0&-i\\
    i&0
  \end{bmatrix},\sigma_z = \begin{bmatrix}
    -1&0\\
    0&1
  \end{bmatrix},
\end{align}
$I$ is the identity quantum operation and $X$ is the quantum \emph{not} gate mathematically represented as
\begin{equation}
  I = \begin{bmatrix}
    1&0\\
    0&1
  \end{bmatrix},X = \begin{bmatrix}
    0&1\\
    1&0
  \end{bmatrix}.
\end{equation}
Moreover, $R^q_z$ and $R^q_y$ can be applied for all $q\in\{1,2,\hdots,n\}$, while CNOT gates $\mathrm{CNOT}^{q_a,q_b}$ are unidirectionally available on adjacent qubit pairs, i.e., $q_b = (q_a + 1) \mod n$. Note that the proposed framework can be easily adjusted for different native gate sets.


\section{Theoretical Framework}\label{sec:theoretical_framework}

In this section, we introduce the theoretical framework derived from the problem of VQE. Then, we equivalently transform the solving procedure into solving a series of multi-objective optimization problems related to gradient magnitudes.

Recall that a VQE task can be considered as minimizing a cost function, which is summarized in {Prob.~\ref{prob:vqe}}.
\begin{problem}[VQE]\label{prob:vqe}
  Given an ansatz $U(\bm{\theta})$ and a task $\mathbb{T}$, the problem of VQE is to find parameters that minimize the cost function $C(\bm{\theta})$, i.e.,
  \begin{equation}
    \min_{\bm{\theta}\in \mathbb{D}^{N_p}} C(\bm{\theta})=\sum_{i} f_i({\rm Tr}\left[O_i U(\bm{\theta}) \rho_i U^\dagger(\bm{\theta})\right]),
  \end{equation}
  where $\mathbb{D} \subset \mathbb{R}$, $N_p$ is the cardinality of trainable parameters $\bm{\theta}$.
\end{problem}

While introducing the VA-VQE framework as summarized in {Prob.~\ref{prob:vavqe}}, the structure of ansatz $U$ is treated as a variable that needs to be optimized in the cost function
\begin{equation}\label{eq:cost_with_ansatz}
  C(U,\bm{\theta})=\sum_{i} f_i({\rm Tr}\left[O_i U(\bm{\theta}) \rho_i U^\dagger(\bm{\theta})\right]).
\end{equation}

\begin{problem}[VA-VQE]\label{prob:vavqe}
  Given a search space of ansatz $\mathbb{S}$ and a task $\mathbb{T}$, the problem of VA-VQE is to find ansatzes and corresponding parameters that minimize the cost function $C(U,\bm{\theta})$, i.e.,
  \begin{align}
      \min_{U\in \mathbb{S},~\bm{\theta}\in\mathbb{D}^{N_p}}C(U,\bm{\theta}),\label{eq:vavqe}
  \end{align}
  where $C(U,\bm{\theta})$ forms as Eq.~(\ref{eq:cost_with_ansatz}), $\mathbb{D} \subset \mathbb{R}$, and $N_p$ is the cardinality of trainable parameters $\bm{\theta}$.
\end{problem}

Practically, the VA-VQE methods automatically construct ansatzes by quantum gates from a given gate set $\mathbb{G}$, i.e. $\mathbb{S} = \mathbb{G}^\infty$, and find ones with trained parameters that minimize the cost function. The search space of ansatz actually scales infinitely since the depth of ansatz circuits can be infinitely large. However, the trainability of the ansatz is substantially limited by the number of quantum gates because the expressibility and the impact of noise of the ansatz may increase with respect to the number of quantum gates, which aggravates the BP phenomenon. On the other hand, it is intractable to search for solutions with enormous (even infinite) search space of ansatz. Therefore, a fixed or gradually increased maximum number $n_g$ of quantum gates is introduced as $\mathbb{S}=\mathbb{G}^{N_g}$, where $N_g$ is sufficiently large that $\exists U \in \mathbb{G}^{N_g}$ such that $\mathbb{U} \cap \mathbb{U}_s \ne \emptyset$. Nevertheless, determining exact solutions requires solving combinatorial optimization in an exponentially large search space, which conceals the efficiency of VQE. Most methods \cite{du2020Quantum,zhang2021Differentiable,bilkis2022Semiagnostic,meng2021Quantuma,chivilikhin2020MoGVQE} attempt to provide approximate solutions in polynomial complexity with respect to $g$. Although the gradient is generally estimated for the optimization of parameters, its magnitude is neglected to provide a guideline for the optimization of the structure of ansatz. Therefore, several undesirable phenomena listed below may occur:
\begin{enumerate}
  \item The ansatz which exhibits severe barren plateau, and thus the gradient magnitude is highly close to $0$, is selected to test the cost function value multiple times.
  \item The trivial ansatz whose parameters are completely trained, i.e. the gradient magnitude is highly close to $0$, is selected as the optimal ansatz since it temporarily cost less than potentially better ansatz with incompletely trained parameters, i.e. the gradient magnitude is significantly larger than $0$.
\end{enumerate}
These phenomena result in the local optimal of the VA-VQE from the ansatz perspective and instability of optimal ansatz and cost function value outputs.

Here we first introduce the gradient magnitude in solving the VA-VQE problem to explicitly supervise the severity of the BP phenomenon and the completeness of parameter training of ansatzes. We define a subproblem the gradient-related ansatz multi-objective optimization (GRAMO) in {Prob.~\ref{prob:multi_obj}}.
\begin{problem}[GRAMO]\label{prob:multi_obj}
  Given a search space of ansatzes $\mathbb{S}$, a task $\mathbb{T}$, and a set of parameters $\mathbb{P}$ such that $\exists! \bm{\theta}_U\in\mathbb{P}$, $\forall U\in\mathbb{S}$, the problem of GRAMO is to find ansatzes that minimize the cost function $C(U,\bm{\theta}_U)$ and maximize the gradient magnitude of the cost function, i.e.,
  \begin{align}\label{eq:multi_obj}
    \min_{U\in \mathbb{S}}C(U,\bm{\theta}_U),~ \max_{U\in \mathbb{S}}\frac{\left\|\nabla C(U,\bm{\theta}_U)\right\|_2}{|\bm{\theta}_U|}
  \end{align}
  where $C(U,\bm{\theta}_U)$ forms as Eq.~(\ref{eq:cost_with_ansatz}), $\nabla C(U,\bm{\theta}_U)$ is the gradient of $C(U,\bm{\theta}_U)$ with respect to $\bm{\theta}_U$, $\bm{\theta}_U \in \mathbb{P}$, and $|\bm{\theta}_U|$ is the cardinality of parameters in $\bm{\theta}_U$.
\end{problem}
The (1-rank) solution of {Prob.~\ref{prob:multi_obj}} is defined as the non-dominated set $\mathbb{U}_n$ consisting of all $U$ such that no $V \in \mathbb{S}$ simultaneously establishes inequalities
\begin{align}
  C(V,\bm{\theta}_{V}) &< C(U,\bm{\theta}_U),\label{eq:cost_ineq}\\
  \frac{\left\|\nabla C({V},\bm{\theta}_{V})\right\|_2}{|\bm{\theta}_{V}|} &\ge \frac{\left\|\nabla C(U,\bm{\theta}_U)\right\|_2}{|\bm{\theta}_U|}.\label{eq:grad_ineq}
\end{align}
As a result, structures of ansatzes with low cost function values or large gradient magnitudes, in other words, high absolute performance or high potentiality, are selected as solutions. Furthermore, the $k$-rank solution $\mathbb{U}_{n}(k)$ of the problem is defined by $\mathbb{U}_{n}(k):=\mathbb{U}_{n}({k-1})\cap\mathbb{U}_{n}^\prime$, where $\mathbb{U}_{n}^\prime$ is the 1-rank solution of {Prob.~\ref{prob:multi_obj}} with search space $\mathbb{S}\setminus \mathbb{U}_{n}^{k-1}$.

We solve the VA-VQE problem by alternately solving {Prob.~\ref{prob:vqe}} and {Prob.~\ref{prob:multi_obj}} as summarized in Alg.~\ref{al:alt_opt}. When the gate set $\mathbb{G}$ and the task $\mathbb{T}$ are specified, for any initialization of $\mathbb{P}$, the VA-VQE can be solved iteratively. At each iteration, the non-dominated set $\mathbb{U}_n$ is determined by solving Prob.~\ref{prob:multi_obj} the GRAMO with $\mathbb{P}$. Then, the optimal parameters $\theta_U \in \mathbb{P}$ are updated via solving Prob.~\ref{prob:vqe} the VQE, $\forall U\in\mathbb{U}_n$. The alternate optimization procedure terminates until $\mathbb{U}_n$ and $\mathbb{P}$ are converged.

\begin{theorem}[Convergence] \label{th:convergence}
  $\mathbb{U}_n$ and $\mathbb{P}$ in Alg.~\ref{al:alt_opt} converge such that $\forall \left(U\in\mathbb{U}_n,\bm{\theta}_U\in\mathbb{P}\right)$ are solutions of Prob.~\ref{prob:vavqe} for all $\mathbb{P}_0$ satisfied that $\exists V \in \mathbb{S}$ such that
  \begin{gather}
    \exists \bm{\theta}^*,~C(V,\bm{\theta}^*) = \min_{U,\bm{\theta}}C(U,\bm{\theta}),\label{eq:exist_optimal_theta}\\
    \forall \delta >0,~C(V,\bm{\theta}_V\in\mathbb{P}_0) \ne \max_{\bm{\theta}\in \Theta_{V,\delta}}C(V,\bm{\theta}),\label{eq:not_local_maximum}
  \end{gather}
  where
  \begin{equation}
    \Theta_{V,\delta} = \{\bm{\theta}\left|\|\bm{\theta}-\bm{\theta}_V\|_2<\delta\right.\}.
  \end{equation}
\end{theorem}

Theorem \ref{th:convergence} (c.f. Appendix \ref{appendix:proof_convergence} for the proof) indicates that solutions of VA-VQE can be determined by Alg.~\ref{al:alt_opt}.

\section{Solving VQE Tasks with Gradient Sensitive Variable Ansatz}\label{sec:practical_implementation}
In this section, the gradient-sensitive alternate framework (\name) for VA-VQE is proposed in detail. As the VA-VQE problem described, \name~takes the gate set $\mathbb{G}$ and the task $\mathbb{T}$ as inputs. Based on the gate commutation rules, the search space of ansatz $\mathbb{S}(N_l)$ under the maximum number of layers $N_l$ is reduced to boost the time efficiency of the optimal ansatz and parameters determination. Then, the quasi-optimal ansatz and corresponding trained parameters are output through three stages: \emph{pool training}, \emph{alternate training}, and \emph{VQE retraining}. In the stage of \emph{pool training}, a candidate tree $T$ is constructed for double $\epsilon$-greedy sampling and a parameter pool in which parameters are shared among ansatzes with similar structures is trained via exploiting the double $\epsilon$-greedy strategy based on the candidate tree. As a result, a reasonable set of parameters $\mathbb{P}_0$ is generated as the initialization of parameters for the next stage. Based on the multi-objective genetic algorithm, the \emph{alternate training} applies the framework described in Alg.~\ref{al:alt_opt} to solve the VA-VQE problem via alternately solving the GRAMO in Prob.~\ref{prob:multi_obj} and the VQE in Prob.~\ref{prob:vqe}. Since the evolutionary algorithm is applied, the output ansatz is quasi-optimal and corresponding parameters may be incompletely trained. Therefore, the third stage \emph{VQE retraining} is required to guarantee the completion of parameter training of the quasi-optimal ansatz. We summarize the \name~as Alg.~\ref{al:gsa}.

Besides, there are several parameters needed to be defined before actually running (c.f. Tab. \ref{tab:hyperparameters} in Appendix \ref{appendix:hyperparameters} for the summarizing). For the entire proposed framework, they are, respectively, the maximum number of layers $N_l$, the reference step size $\alpha_0$, the convergence threshold $\xi$, and the probabilities of greedy selection for the double $\epsilon$-greedy strategy $\epsilon_1$ and $\epsilon_2$. In the \emph{pool training}, they are, respectively, the number of sampled ansatzes $N_{s1}$, the maximum number of ranks of ansatzes whose corresponding parameters are updated $N_{r1}$, the stable threshold for terminating the main process $N_{t1}$, the maximum iteration times in the prethermalization $N_{i0}$, and the maximum iteration times in the main process $N_{i1}$. In the \emph{alternate training}, they are, respectively, the population size $N_{s2}$, the maximum number of ranks of ansatzes whose corresponding parameters are updated $N_{r2}$, the stable threshold $N_{t2}$, the optimization step in a generation $N_{o}$, and the maximum iteration times $N_{i2}$. In the \emph{VQE retraining}, they include the maximum iteration times $N_{i3}$.

This section is organized as follows: In Sec.~\ref{subsec:search_space}, we concretely propose the structure of the search space of ansatz for the \name and conduct the size reduction to the search space based on the gate commutation rules. Subsequently, the \emph{pool training}, \emph{alternate training}, and \emph{VQE retraining} stages are explained in 
Sec.~\ref{subsec:pool_training}, \ref{subsec:alternate_training} and \ref{subsec:vqe_retraining}, respectively.

\begin{figure}[t]
  \centering
	\includegraphics[width=0.5\textwidth]{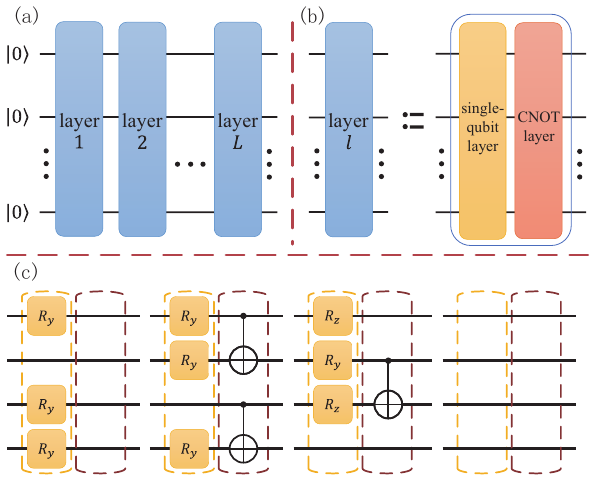}
  \caption{\label{fig:layer} The layer-by-layer ansatz (a) with each layer being decomposed into two sublayers, i.e., (b) a single-qubit gate layer and a CNOT layer; and (c) four exemplary constructions of layers. 
  }
\end{figure}

\subsection{Search Space of Ansatz \label{subsec:search_space}}
Without violation of Eq.~\ref{eq:ansatz}, we compose an ansatz in a layer-by-layer fashion as shown in Fig.~\ref{fig:layer}. To be specific, each layer consists of a set of native disjoint single-qubit gates followed by a set of hardware-compliant disjoint CNOT gates. It can be easily derived that the number of all possible structures $N_s$ of a layer is exponentially large with respect to the number of qubits $n$. We use the term state denoted by $\astate_j$ to indicate the $j$th structure is used. Then, an ansatz can be identified by a tuple of states 
\[
  \left(\astate^1, \astate^2, \hdots, \astate^{N_l}\right),
  \]
where $\astate^l \in S := \left\{\astate_1, \astate_2, \hdots, \astate_{N_s}\right\}$. As a result, the search space of ansatz $\mathbb{S}(N_l)$ can be specifically represented by a tree with $N_l + 1$ layers (from layer $0$ to $N_l$). A node at the layer $l>0$ in the tree is uniquely identified by its parent node and the state, i.e., $\anode^l := (\anode^{l-1}, \astate^l)$, and linked to its child nodes. In layer $0$, the root node of the tree is defined with no parent node and the state $\astate^0$ in which no layer information is stored. Furthermore, we call the nodes in layer $N_l$ linked to no child nodes the leaf nodes. Finally, a path from the root node to a leaf node $\left(\anode^0, \anode^1,\hdots,\anode^{N_l}\right)$ represents an ansatz $\left(\astate^1, \astate^2, \hdots, \astate^{N_l}\right)$.

The unrefined $\mathbb{S}(N_l)$ is intuitively a full $N_s$-ary tree with $(N_s)^{N_l}$ leaf nodes. However, this simple-minded construction can not establish the bijection between practical implementations of ansatzes and paths (c.f. Eg.~\ref{eg:not_bijection}). Moreover, gates in distinct layers may be deletable or mergeable (c.f. Eg.~\ref{eg:mnd}). To further eliminate the redundant ansatzes and improve the efficiency, the paths violating the following cross-layer constraints based on gate commutation rules will be pruned. 

\paragraph{Constraints:} For node $\anode^l:=(\anode^{l-1},\astate^l)$ in the path $(\anode^0,\anode^1, \hdots,\anode^l)$:
\begin{enumerate}
  \item $R_z^q,\mathrm{CNOT}^{q,q^\prime}\notin\astate^l$, if $R_y^q$ and $\mathrm{CNOT}^{q^{\prime\prime},q}$ $\notin\astate^k$, $\forall q$ with $\left|0\right>$ initialization, $q^\prime$ and $ q^{\prime\prime}$, $l>k>0$;
  \item $\mathrm{CNOT}^{q_1,q_2}\notin\astate^l$, if $\exists \mathrm{CNOT}^{q_1,q_2}\in \astate^k$ such that $R_y^{q_1}$, $R_y^{q_2}$, $R_z^{q_2}$, $\mathrm{CNOT}^{q^\prime,q_1}$, $\mathrm{CNOT}^{q_2, q^\prime}$, and $\mathrm{CNOT}^{q^\prime,q_2}$ $\notin \astate^m$, $\forall q, q^\prime$, $k<m<l<0$;
  \item $R_z^q\notin\astate^l$, if $R_y^q$ and $\mathrm{CNOT}^{q^{\prime},q}$ $\notin \astate^{l-1}$, $\forall q, q^\prime$, $l>0$;
  \item $R_y^q\notin\astate^l$, if $R_z^q$, $\mathrm{CNOT}^{q,q^{\prime}}$, and $\mathrm{CNOT}^{q^{\prime},q}\notin \astate^{l-1}$, $\forall q, q^\prime$, $l>1$.
  \item $\astate^l$ is empty if $\astate^{l-1}$ is empty, $l>1$.
\end{enumerate}

The first constraint follows the fact that $R_z^q$ and $\mathrm{CNOT}^{q,q\prime}$ preserve the state of the quantum system when the quantum state of $q$ is $\left|0\right>$. The second constraint avoids two consecutive CNOT gates with identical control and target qubits. As a result, the first two constraints eliminate deletable combinations of quantum gates. The mergeable combinations are extinguished by constraints 3 and 4 by prohibiting consecutive $R_z$ and $R_y$, respectively. Finally, the bijection between practical implementations of ansatzes and paths is established by the conjunction of constraints 3 to 5. Consequently, the size of the search space $\left|\mathbb{S}(N_l)\right|$ is significantly reduced. We provide an example Eg.~\ref{eg:space_size} to demonstrate the efficiency obtained from our constraints.

\subsection{Pool training \label{subsec:pool_training}}
In this subsection, the \emph{pool training} stage is presented in detail. Based on the weight-sharing policy, the number of parameters required to be trained is reduced to $\mathcal{O}(N_l)$. Moreover, the double $\epsilon$-greedy strategy is exploited accompanied by a candidate tree to mitigate the training competitions among ansatzes. We treat the \emph{pool training} as a one-shot training program. At each iteration, several ansatzes are randomly sampled and estimated. According to the estimated performance, several temporarily outstanding ansatzes are selected to update parameters in specified steps (typically one step) via a gradient-based optimizer.

Recall that solving VQ-VQE in the framework as described in Alg.~\ref{al:alt_opt} requires the initialization of $\mathbb{P}_0$ in which each structure of ansatz links to an independent set of parameters. Intuitively, that parameters $\bm{\theta}_U \in \mathbb{P}_0$ substantially reflect the actual performance of ansatz $U$, i.e., $\bm{\theta}_U \approx \arg\min_{\bm{\theta}}C(U,\bm{\theta})$, facilitates the optimal ansatz determination. Therefore, we conduct a pre-training of parameters before the alternate optimization solving VA-VQE instead of random initialization. Unfortunately, the size of $\mathbb{P}_0$ is exponentially large with respect to $n$ and $N_l$ resulting from the exponentially large $\mathbb{S}(N_l)$. It is impractical to adequately train $\mathbb{P}_0$ efficiently. 

Instead, we construct a parameter pool with linear size with respect to $N_l$ via applying the weight-sharing policy and train the pool to eventually derive $\mathbb{P}_0$. The parameter pool can be matrix-like defined as
\begin{equation}\label{eq:theta_expand}
  P := \begin{bmatrix}
    \bm{\theta}_{\astate_1,1}& \bm{\theta}_{\astate_1,2}& \hdots& \bm{\theta}_{\astate_1,N_l}\\
    \bm{\theta}_{\astate_2,1}& \bm{\theta}_{\astate_2,2}& \hdots& \bm{\theta}_{\astate_2,N_l}\\
    \vdots &\vdots& \ddots& \vdots\\
    \bm{\theta}_{\astate_{N_s},1}& \bm{\theta}_{\astate_{N_s},2}& \hdots& \bm{\theta}_{\astate_{N_s},N_l}\\
  \end{bmatrix},
\end{equation}
where $\bm{\theta}_{k,l}$ represents the parameters at the $l$th layer corresponding to the state $\astate_k$. Then, parameters of an ansatz $\left(\anode^0, \anode^1,\hdots,\anode^{N_l}\right)$ is 
\begin{equation}
  \bm{\theta} := \bm{\theta}_{\astate^1,1} \oplus \bm{\theta}_{\astate^2,2} \oplus \hdots \oplus \bm{\theta}_{\astate^{N_l},N_l},
\end{equation}
where $\oplus$ indicates the direct sum such that 
\[
  \begin{bmatrix}
    a_1\\ a_2\\ \vdots \\a_N
  \end{bmatrix} \oplus \begin{bmatrix}
    b_1\\ b_2\\ \vdots \\b_N
  \end{bmatrix} = \begin{bmatrix}
    a_1\\ \vdots \\a_N \\ b_1 \\ \vdots \\b_N
  \end{bmatrix}.
\]
It can be simply derived that parameters in identical layers and states are shared among ansatzes, which is the direct effect of applying the weight-sharing policy. Then, the number of trainable parameters is in $\mathcal{O}({N_sN_l})$, which is linear with respect to $N_l$. The exponentially reduction significantly boosts the efficiency of parameter training. However, the training competitions are therefore introduced.

Inspired by the $\epsilon$-greedy strategy from traditional machine learning \cite{you2020Greedynas}, we propose the double $\epsilon$-greedy strategy based on the candidate tree to mitigate the training competitions. Remarkably, our proposed method not only differentiates `good' and `bad' ansatzes but also ansatzes among the two categories.

Recall that the search space $\mathbb{S}(N_l)$ can be represented by a tree. It is intuitive that constructing a tree to save potentially `good' ansatzes and discarding potentially `bad' ansatzes are reasonable for the differentiation. The tree spanned by paths representing potentially `good' ansatzes is named the candidate tree. Each node $\anode$ in the tree maintains a {leaf count} $c_l(\anode)$ to indicate the number of leaf nodes below $\anode$ and a {training count} $c_t(\anode)$ to record the total times of training of the node. At each path sampling procedure, the \name~samples a path from the candidate tree with the probability $\epsilon_1$ and the $\mathbb{S}(N_l)$ uniformly with the probability $1-\epsilon_1$. While sampling from the candidate tree, nodes are successively sampled from the root to a leaf. Since only potentially `good' ansatzes are appended to the candidate tree and recorded the training count of nodes via the one-shot training scheme, nodes in the candidate tree with a large training count may intuitively have more probability to construct potentially `good' ansatzes. Let the last selected node be $\anode^{l-1}$, $1\le l\le N_l$, linked to child nodes $\anode_1^{l},\anode_2^{l},\hdots,\anode_{N_\anode}^{l}$. Then, the next node is sampled as $\anode_k^l$ with the probability 
\begin{equation}
  {\rm Pr}(\anode^l_k;\eta) = \frac{c_l(\anode^l_k) + \eta c_t(\anode^l_k)}{\sum_{i=1}^{N_{\anode}}c_l(\anode^l_i) + \eta c_t(\anode^l_i)},
\end{equation}
where $\eta=1$ with the probability $\epsilon_2$ representing the greedy sampling and $\eta=0$ with the probability $1-\epsilon_2$ representing the uniform sampling. The ansatz sampling procedure is summarized as Alg.~\ref{al:sample}.

As depicted in Fig.~\ref{fig:pool_training}, at each main process iteration of the \emph{pool training}, the \name~samples $N_{s1}$ paths exploiting the double $\epsilon$-greedy strategy. Then, the cost function values and gradients of sampled ansatzes are estimated with corresponding parameters from the parameter pool. Subsequently, the $N_{r1}$-rank solution $\mathbb{U}_n({N_{r1}})$ of Prob.~\ref{prob:multi_obj} is determined among the sampled ansatzes. For each ansatz $U\in \mathbb{U}_n({N_{r1}})$, corresponding parameters in the parameter pool are updated in one step via gradient descent optimizer. The main process terminates at iteration $N_{i1}$ or when $c_l(\anode^0)$ is stable that the value has been preserved for $N_{t1}$ iterations which means there is no new path appended on the tree. As a result, the main process of the \emph{pool training} is summarized as Alg.~\ref{al:main_pool_training}.

\begin{figure}[t]
  \centering
	\includegraphics[width=0.5\textwidth]{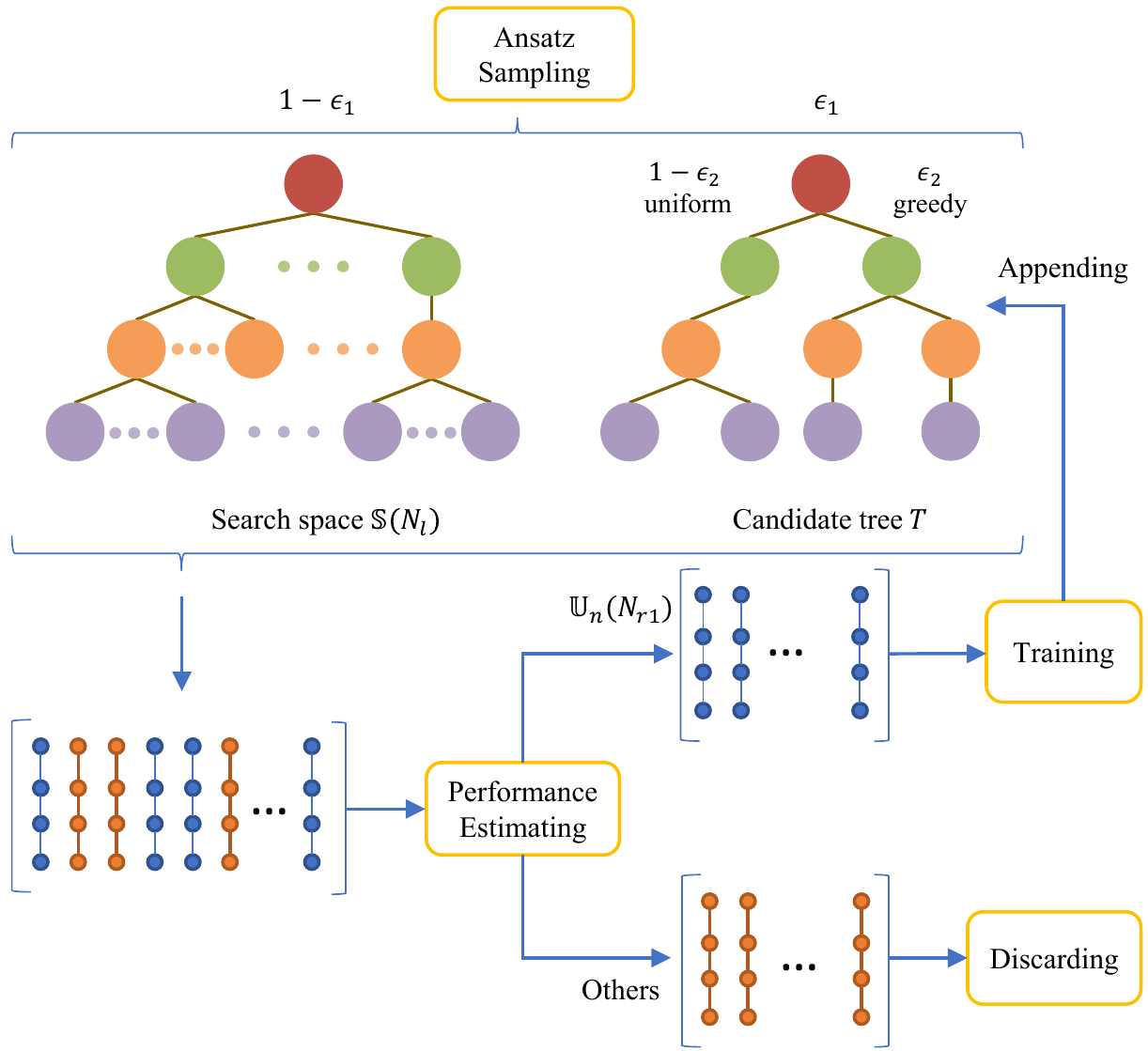}
  \caption{\label{fig:pool_training} Single iteration of \emph{pool training} with $N_l=3$.
  }
\end{figure}

For the stability of \emph{pool training}, we provide a prethermalization before the main process. Given an $N_{i0}$, at iteration $i$, the \name~processes as the main process iteration with $\epsilon_1^\prime = {(i-1)\epsilon_1}/{N_{i0}}$ instead of $\epsilon_1$. The prethermalization of \emph{pool training} is summarized as Alg.~\ref{al:pre_pool_training}.

Finally, the \emph{pool training} can be described by Alg.~\ref{al:pool_training}. After the initialization of $P$ and $T$, the prethermalization and the main process are conducted to train $P$ as well as $T$. Subsequently, the set of parameters $\mathbb{P}_0$ can be constructed by expanding $P$ that $\bm{\theta}_U$ is generated by Eq.~(\ref{eq:theta_expand}) for any ansatz $U\in\mathbb{S}(N_l)$. The $\mathbb{P}_0$ and $T$ are output for the next stage.

\subsection{Alternate training \label{subsec:alternate_training}}
Recall that the VA-VQE can be solved via alternately solving Prob.~\ref{prob:multi_obj} and Prob.~\ref{prob:vqe}. We exploit the multi-objective genetic algorithm with novel modification. The set of parameters $\mathbb{P}$ is initialized as $\mathbb{P}_0$. The individuals in the first generation are sampled independently via the double $\epsilon$-greedy strategy as described in Sec.~\ref{subsec:pool_training} to compose the initial population. 

\begin{figure}[t]
  \centering
	\includegraphics[width=0.5\textwidth]{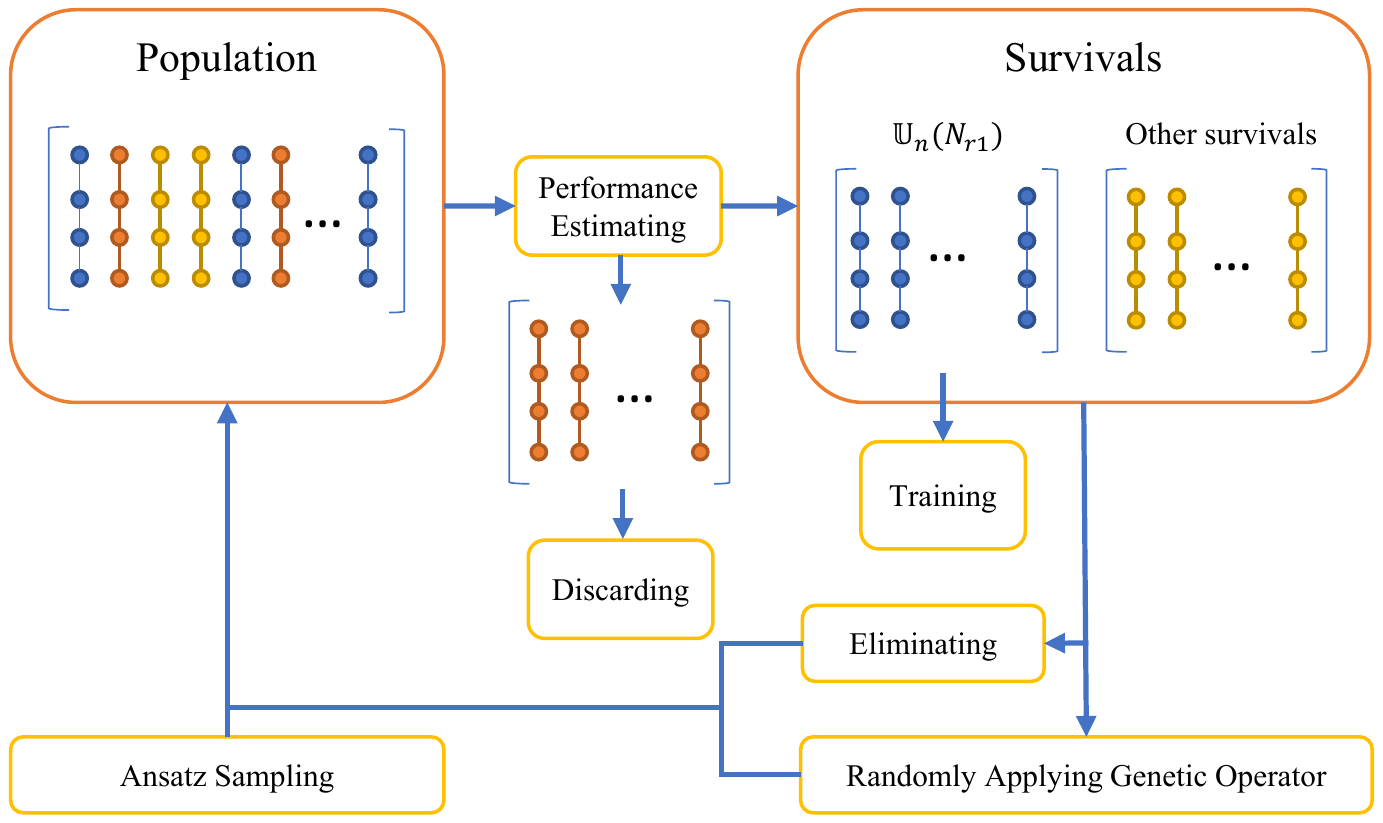}
  \caption{\label{fig:alt_training} Single generation of \emph{alternate training}.
  }
\end{figure}

As shown in Fig.~\ref{fig:alt_training}, at each generation (iteration), ansatzes in $\mathbb{U}_n(N_{r2})$ from solving Prob.~\ref{prob:multi_obj} with respect to the population are trained to update $\mathbb{P}$ in $N_{o}$ steps via gradient descent optimizer. Note that the number of ansatzes in $\mathbb{U}_n(N_{r2})$ should be less than $N_{s2}/2$. Similar to the traditional genetic algorithm NGSA-II, $N_{s2}/2$ ansatzes are survived. Specifically, the \name~finds an $N_r^\prime$ such that 
\[\left|\mathbb{U}_n(N_r^\prime)\right|<N_{s2}/2\le \left|\mathbb{U}_n(N_r^\prime+1)\right|.\]
Then, ansatzes are sequentially inserted to survivals in increasing order of the cost function value from $\mathbb{U}_n(N_r^\prime+1)\setminus \mathbb{U}_n(N_r^\prime)$. Subsequently, new ansatzes are generated by applying asexual genetic operators on survivals to fill the population. We leave the description of asexual genetic operators in Appendix \ref{appendix:genetic_operator}. Besides, we introduce the explicit elimination to the near-completely trained survivals $U$ satisfying $\alpha\frac{\left\|\nabla C(U,\bm{\theta}_U)\right\|_2}{|\bm{\theta}_U|}<\xi$, where $\alpha$ is the step size. The eliminated ansatz is recorded if it has the temporarily lowest cost function value and is erased if there exists an ansatz with a lower cost function value. Notice that the elimination is conducted after the application of genetic operators. Therefore, the sampling based on the double $\epsilon$-greedy is required to refill the population.

The \emph{alternate training} terminates at the generation $N_{i2}$ or when the record eliminated ansatz is preserved for $N_{t2}$ generations. As a result, the ansatz $U^*$ with the temporarily lowest cost function value is output as the quasi-optimal ansatz. Meanwhile, the corresponding parameters $\bm{\theta}_{U^*}$ are output for the next stage as the parameter initialization. We summarize the \emph{alternate training} in Alg.~\ref{al:alt_training}.

\subsection{VQE retraining \label{subsec:vqe_retraining}}

The final stage of \emph{VQE retraining} inherits the quasi-optimal ansatz $U^*$ and corresponding parameters $\bm{\theta}_{U^*}$ output from the \emph{alternate training} and provides a guarantee of the sufficiency of the parameter training of $\bm{\theta}_{U^*}$. As traditional VQE training does, this stage simply trains the parameters of $\bm{\theta}_{U^*}$ with the initialization $\bm{\theta}_{U^*}$ until the cost function value converges or the iteration count reaches $N_{i3}$. We summarize this stage as Alg.~\ref{al:vqe_retraining} for completeness.

\section{Numerical Simulations}\label{sec:simulations}

\begin{figure*}[t]
  \centering
	\subfigure[Ground State Energy]{
		\begin{minipage}[h]{0.3\textwidth}
		\centering
		\includegraphics[width=\textwidth]{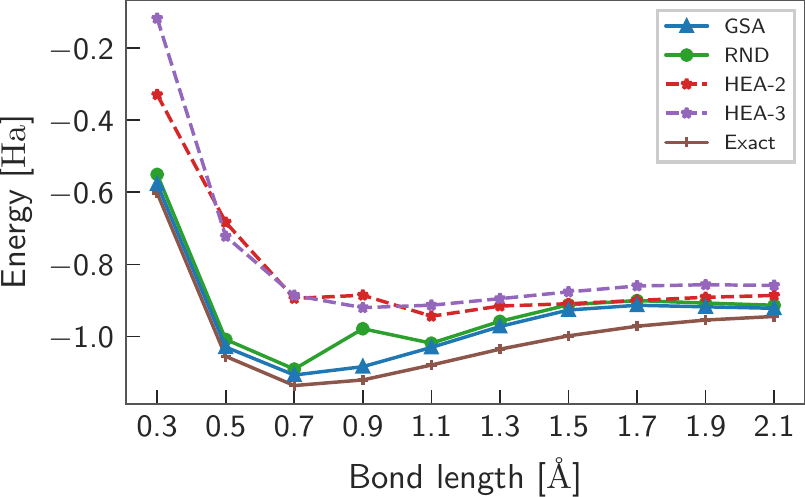}
		\end{minipage}
	}
  \subfigure[Absolute Error]{
		\begin{minipage}[h]{0.3\textwidth}
		\centering
		\includegraphics[width=\textwidth]{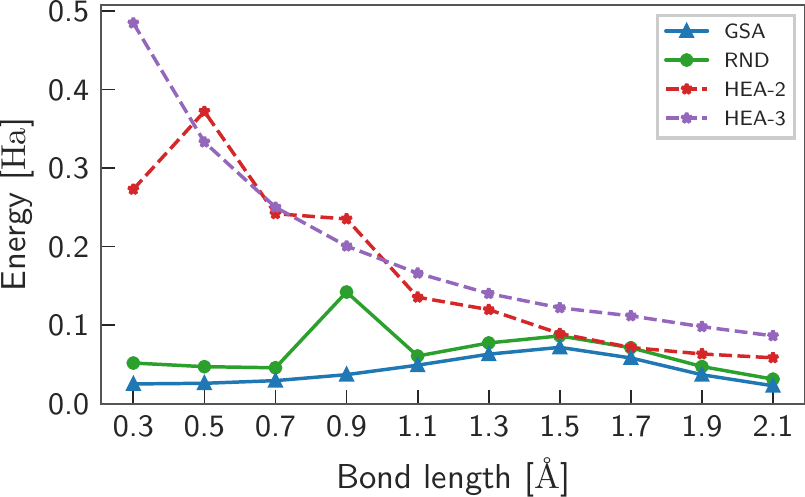}
		\end{minipage}
	}
  \subfigure[Quantum Cost]{
		\begin{minipage}[h]{0.3\textwidth}
		\centering
		\includegraphics[width=\textwidth]{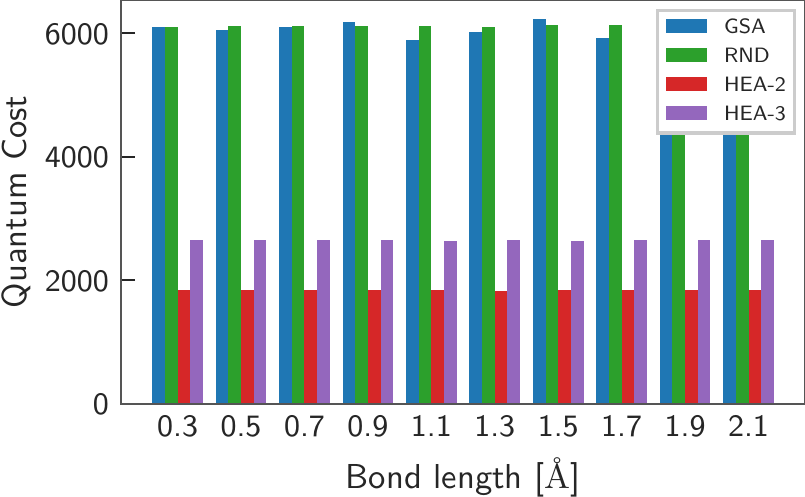}
		\end{minipage}
	}
  \caption{\label{fig:res_bond_ave} The average result of determining ground state energy of $\mathrm{H}_2$ in various bond lengths among 100 running times.}
\end{figure*}

\begin{figure*}[t]
  \centering
	\subfigure[Ground State Energy]{
		\begin{minipage}[h]{0.3\textwidth}
		\centering
		\includegraphics[width=\textwidth]{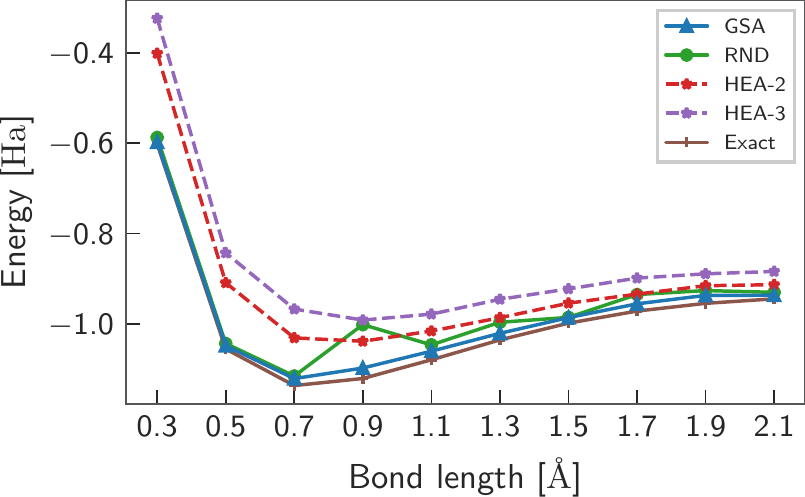}
		\end{minipage}
	}
  \subfigure[Absolute Error]{
		\begin{minipage}[h]{0.3\textwidth}
		\centering
		\includegraphics[width=\textwidth]{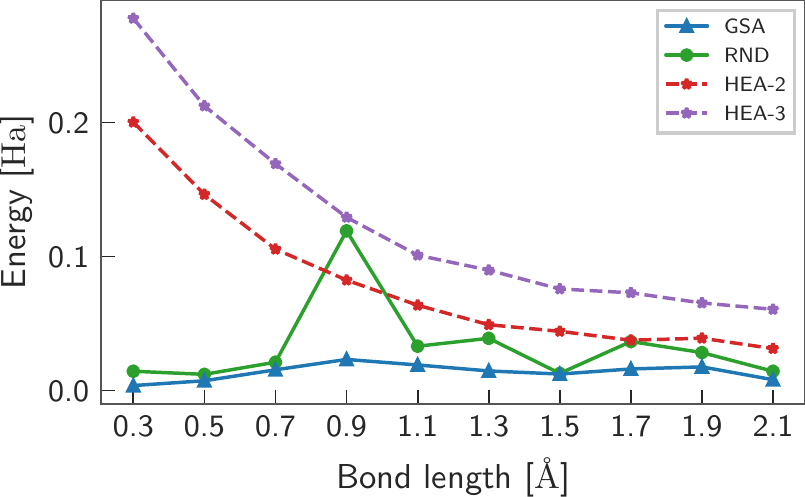}
		\end{minipage}
	}
  \subfigure[Quantum Cost]{
		\begin{minipage}[h]{0.3\textwidth}
		\centering
		\includegraphics[width=\textwidth]{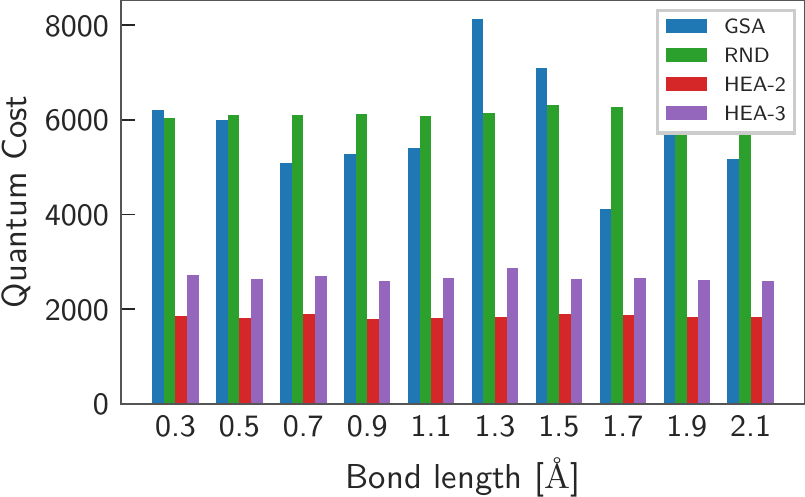}
		\end{minipage}
	}
  \caption{\label{fig:res_bond_best} The best result of determining ground state energy of $\mathrm{H}_2$ in various bond lengths among 100 running times.}
\end{figure*}

In this section, we conduct numerical simulations on VQE tasks for finding ground state energies of $\mathrm{H}_2$ and $\mathrm{H}_4$ molecules to show the improvement of the proposed framework. Moreover, to demonstrate the flexibility, we implement a modified version of the proposed framework for meta-VQE \cite{cervera-lierta2021Metavariational} learning energy profiles of parameterized Hamiltonians. 

\subsection{Finding ground state energies}

To showcase the improvement of the proposed framework, we compare it with structure fixed hardware efficient ansatz (HEA) \cite{kandala2017Hardwareefficient} and a fully randomized (RND) baseline algorithm. To make a fair comparison, we use Python with the \emph{Pennylane} package \cite{bergholm2020Pennylane} to implement both our and the two compared methods. The technical details for algorithms to be compared are elaborated as follows:

\begin{itemize}
  \item \textbf{HEA.} HEA utilizes a fixed layer pattern to construct the whole circuit. The layer construction used in this implementation is depicted in Appendix \ref{appendix:hea}. Note that this kind of ansatzes is also widely used in other implementations, e.g., \cite{meng2021Quantum,li2022Quantum};
  \item \textbf{RND.} We sample $N_{Rs}$ circuits with $N_l$ layers and unique random parameters, then process the VQE retraining and output the circuit with the minimum cost function value.
\end{itemize}

We assume that depolarization dominates the quantum error in the quantum hardware. Then, the noise operator (depolarizing channel) is 
\begin{equation*}
  \label{eq:depolarizing}
  \begin{gathered}
    \Phi(p,\sigma)=\sum_{i=0}^{3}K_i(p) \sigma K_i^\dagger(p),\\
    K_0 = \sqrt{1-p}
      \left[
        \begin{matrix}
          1&0\\
          0&1
        \end{matrix}
      \right], K_1 = \sqrt{\frac{p}{3}}
      \left[
        \begin{matrix}
          0&1\\
          1&0
        \end{matrix}
      \right],
    K_2 = \sqrt{\frac{p}{3}}
      \left[
        \begin{matrix}
          0&-i\\
          i&0
        \end{matrix}
      \right],
      K_3 = \sqrt{\frac{p}{3}}
      \left[
        \begin{matrix}
          1&0\\
          0&-1
        \end{matrix}
      \right],
  \end{gathered}
\end{equation*}
where $p\in[0,1]$ is the depolarization probability and $\sigma$ is a density matrix of a single qubit \cite{bergholm2020Pennylane}. In our configuration, a depolarizing channel with depolarization probability $p=0.001$ is applied after a single-qubit gate. We simultaneously apply two depolarizing channels with depolarization probability $p=0.01$ after a CNOT, i.e., apply one on the control qubit, and one on the target qubit.

\begin{figure}[t]
  \centering
	\includegraphics[width=0.45\textwidth]{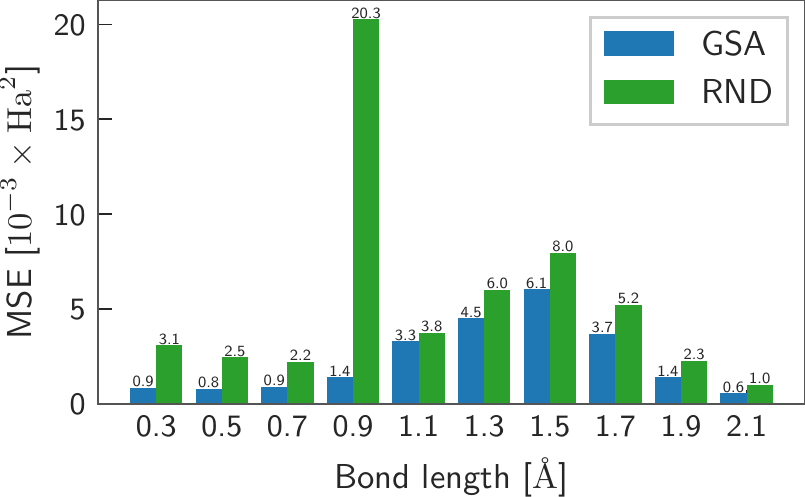}
  \caption{\label{fig:mse} Mean square error among 100 running times.
  }
\end{figure}

The criteria used for all methods are absolute error between the obtained and exact optimal cost function values, and the invoking times for the calculation of the cost function (termed {\emph{quantum cost}} henceforth), respectively. For VA-VQE, we analyze the distribution of performances of output ansatzes and calculate the mean square error (MSE) mathematically represented by
\begin{equation}
  MSE = \frac{1}{M}\sum_{i=1}^M(\hat{C}^*_i-C_{exact}^*)^2
\end{equation}
to indicate the stability to obtain the quasi-optimal cost function values, where $M$ is the total running time, $\hat{C}^*$ is the output quasi-optimal cost function value and $C_{exact}^*$ is the exact optimal cost function value.

\begin{figure}[t]
  \centering
	\subfigure[]{
		\begin{minipage}[h]{0.45\textwidth}
		\centering
		\includegraphics[width=\textwidth]{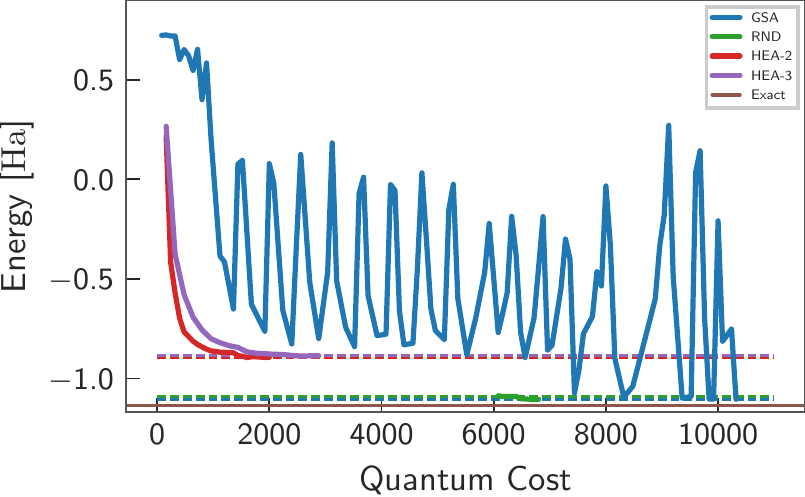}
		\end{minipage}
	}
  \subfigure[]{
		\begin{minipage}[h]{0.45\textwidth}
		\centering
		\includegraphics[width=\textwidth]{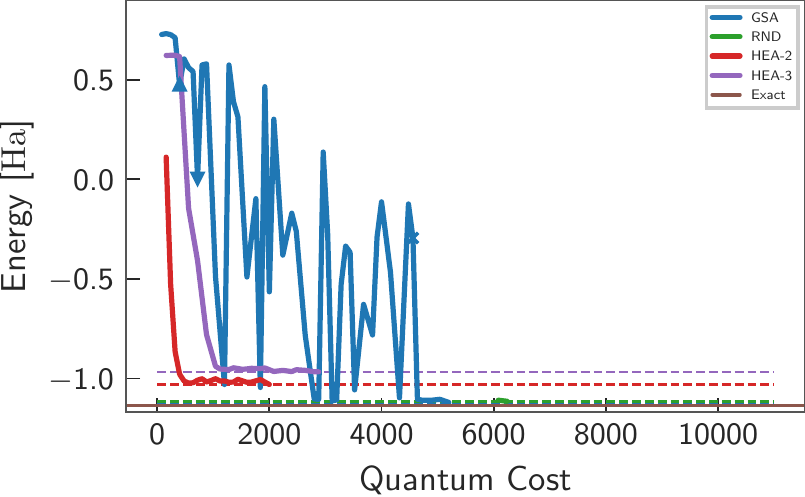}
		\end{minipage}
	}
  \caption{\label{fig:H2_07_detail_ave} Detail of cost function values with respect to quantum cost in solving the ground state energy of $\mathrm{H}_2$ in bond length $0.7$\AA~among 100 running times. (a) The average cost function values with respect to quantum cost. (b) The cost function values among 100 running times with respect to quantum cost in the running time which output the ansatz with the lowest cost function value among 100 running times. The $\blacktriangle$, $\blacktriangledown$, and $\times$ represent the termination of \emph{pool training}'s prethermalization, \emph{pool training}'s main process, and \emph{alternate training}, respectively.}
\end{figure}

\begin{figure*}[t]
  \centering
	\subfigure[$\mathrm{H}_2$ in bond length $0.7$\AA]{
    \label{subfig:sca_H2}
		\begin{minipage}[h]{0.45\textwidth}
		\centering
		\includegraphics[width=\textwidth]{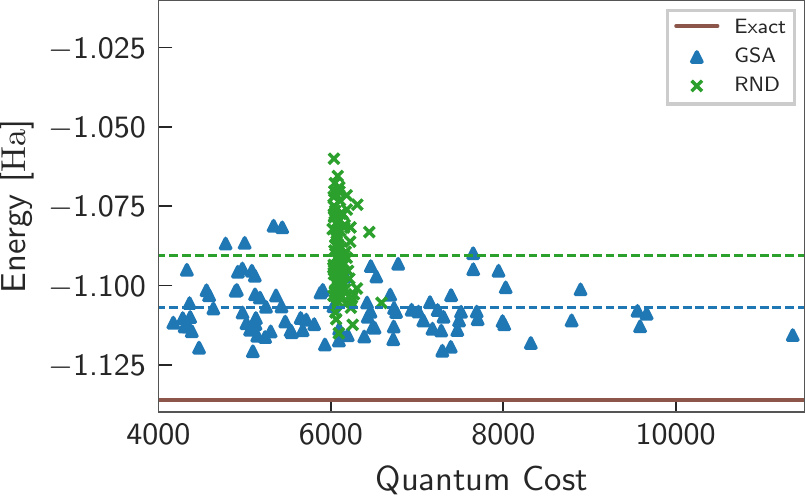}
		\end{minipage}
	}
  \subfigure[$\mathrm{H}_4$ in bond length $0.7$\AA]{
    \label{subfig:sca_H4}
		\begin{minipage}[h]{0.45\textwidth}
		\centering
		\includegraphics[width=\textwidth]{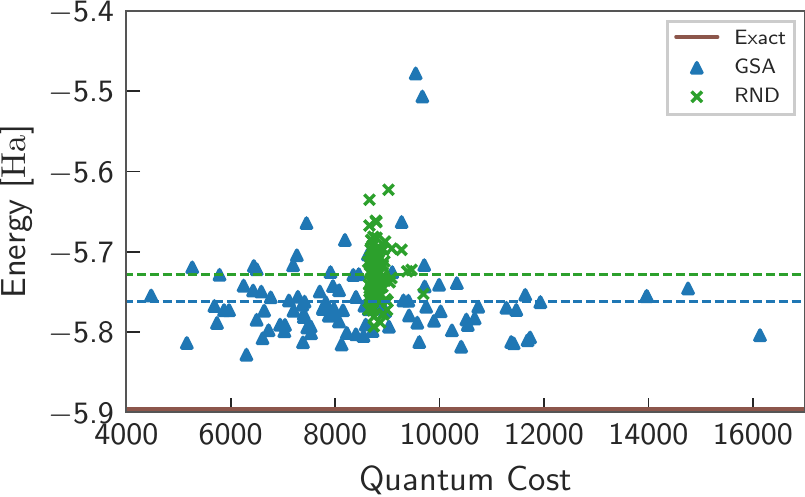}
		\end{minipage}
	}
  \caption{\label{fig:sca} Output quasi-optimal cost function values with respect to the quantum cost of 100 running times in solving the ground state energy. Dashed lines present cost function values on average.}
\end{figure*}

We first conduct numerical simulations on determining the ground state energies of Hamiltonians of $\mathrm{H}_2$ in various bond lengths. The quantum device is assumed to have 4 qubits with {ring} connectives and depolarization error being imposed by the classical simulator. For the proposed framework, we empirical set $N_l=3$, $\alpha_0=5$, $\xi =0.003$, $\epsilon_1=\epsilon_2=0.8$, $N_{s1} = N_{s2} = 16$, $N_{r1}=1$, $N_{r2}=2$, $N_{i0}=N_{i1}=2$, $N_{i2}=100$, $N_{i3} = 10$, $N_{t1} = 1$, $N_{t2} = 4$, and $N_o=5$; for HEA, we exploit 2-layer HEA as HEA-2 and 3-layer HEA as HEA-3; for RND, we set $N_l=3$ and $N_{Rs} = 6000$ such that the quantum costs of \name~and RND are approximately equal.

For each bond length, we conduct the \name~and comparison methods $100$ times, respectively. The results on average are shown in Fig.~\ref{fig:res_bond_ave}. It can be found that our algorithm can determine ground state energy with consistently lower absolute error on average than the compared algorithms in terms of the absolute error in various bond lengths, and with lower quantum cost compared to RND. This conclusion holds while considering the ansatzes with the lowest cost function values among the $100$ running times as depicted in Fig.~\ref{fig:res_bond_best}. Moreover, \name~exhibits outstanding stability as shown in Fig.~\ref{fig:mse}. Remarkably, VA-VQE methods performs better in terms of absolute error at a cost of quantum cost than HEA in most bond lengths.

For bond length $0.7$\AA, we also record the quantum cost and absolute error for each running time in Fig.~\ref{fig:H2_07_detail_ave} and Fig.~\ref{subfig:sca_H2} when the algorithm converges to the final output result. It can be concluded that our algorithm can obtain the best solution on average with a rather small quantum cost (about 6107.49) which is conspicuously better than HEA (up to 87.9\% improvement in terms of error) and RND (up to 36.0\% and 58.7\% improvement in terms of error and stability, respectively). Remarkably, \name~can asymptotically reach the best ansatz among the training, which means that a temporarily quasi-optimal ansatz can be output at any generation of \emph{alternate training}. However, our framework shows large variance of the quantum cost. This may result from the termination conditions of the three stages. It can be concluded that the termination conditions prevent the redundant iterations which can not significantly benefit the training, and thus the performance of our framework may not be sensitive to hyperparameters $N_{i0}$, $N_{i1}$, $N_{i2}$, and $N_{i3}$ when they are sufficiently large.

\begin{table}[htp]
  \begin{center}
    \caption{The result of determining the ground state energy of $\mathrm{H}_4$ in bond length $0.7$\AA. Data outside and inside brackets represent the average and the best values among 100 running times, respectively.}
    \label{tab:res_H4}
    \begin{tabular}{c| c c c c}
    &\name& RND& HEA-2& HEA-3\\
    \hline
    energy [$\mathrm{Ha}$]& -5.762(-5.828)& -5.728(-5.793)& -5.350(-5.487)& -5.230(-5.439)\\
    absolute error [$\mathrm{Ha}$]& 0.134(0.068)& 0.167(0.103)& 0.546(0.408)&0.666(0.457)\\
    error improvement& -& 20.1\%(34.1\%)& 75.5\%(83.4\%)& 79.9\%(85.2\%)\\
    quantum cost& 8579(6301)& 8811(8741)& 2665(2695)& 3865(3887)\\
    stability &0.020& 0.029& -& -\\
    \end{tabular}
  \end{center}
  \end{table}


We also conduct 6-qubit simulations for $\mathrm{H}_4$ in bond length $0.7$\AA. Configurations of \name~and HEA are set as them for $\mathrm{H}_2$. As for RND, we set $N_{Rs} = 8600$. Results of the simulations are shown in Tab.~\ref{tab:res_H4} and Fig.~\ref{subfig:sca_H4}. Remarkably, the proposed method still shows significant improvements in terms of absolute error by up to 75.5\% on average compared to HEA and both the absolute error and stability by up to 20.1\% and 29.5\%, respectively, compared to RND.

\subsection{Learning energy profiles}

To showcase the flexibility, we implement a modified version of the proposed framework for Meta-VQE \cite{cervera-lierta2021Metavariational}. Let the parameterized Hamiltonian be $H(\Delta)$. Then, sample $M$ available values of $\Delta$ denoted by $\Delta_1$, $\Delta_2$, $\hdots$, $\Delta_M$ as training bonds. Hence, the cost function for training Meta-VQE is defined by
\begin{equation}
  C_{\mathrm{Meta}}(\bm{\theta}) = \sum_{i=1}^M {\rm Tr}\left[H(\Delta_i) U(\bm{f}(\Delta_i;\bm{\theta})) \rho_i U^\dagger(\bm{f}(\Delta_i;\bm{\theta}))\right],
\end{equation}
where $\bm{f}(\Delta_i;\cdot)$ represents the encoding function of the parameter of the Hamiltonian with respect to operand trainable parameters. After parameter training, the landscape of
\begin{equation}
  V_{\mathrm{Meta}}(\Delta)={\rm Tr}\left[H(\Delta) U(\bm{f}(\Delta_i;\bm{\theta})) \rho_i U^\dagger(\bm{f}(\Delta;\bm{\theta}))\right]
\end{equation}
exhibits the energy profile of the parameterized Hamiltonian.

We implement the Meta-VQE adapted \name~by simply identifying each state $\astate$ by both the structure and the encoding function of each single-qubit gate in the structure to enable the auto-decision of the encoding methods. Here we provide 3 encoding functions for parameters of single-qubit gates. They are, respectively, 
\begin{align}
  &f_1(\Delta;\theta)=\theta,\\
  &f_2(\Delta;\theta, \gamma) = \theta\Delta+\gamma,\\
  &f_3(\Delta;\theta, \gamma) = \theta e^\Delta+\gamma.
\end{align}

We compare our framework with the HEA implementation. Each HEA consists of $N_{le}$ encoding layers and $N_{lp}$ processing layers and is denoted by HEA-$N_{le}$-$N_{lp}$. Specifically, $f_2$ and $f_1$ are applied on each single-qubit gate in encoding layers and processing layers, respectively.

\begin{figure*}[t]
  \centering
	\subfigure[Energy Profile]{
		\begin{minipage}[h]{0.45\textwidth}
		\centering
		\includegraphics[width=\textwidth]{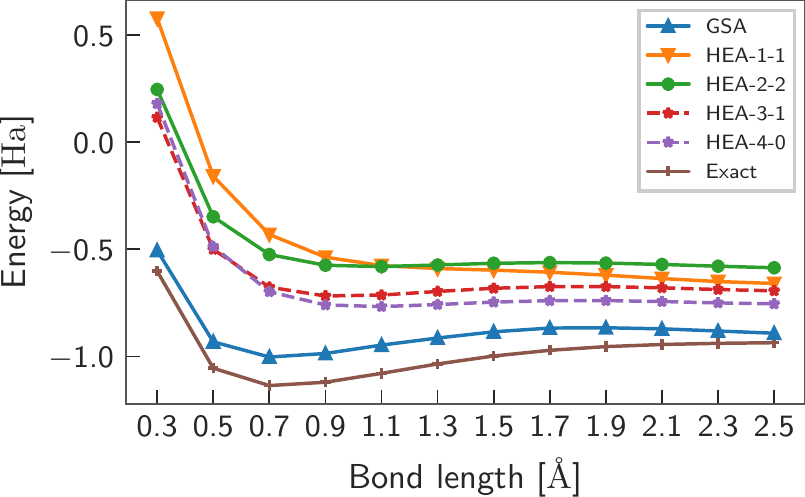}
		\end{minipage}
	}
  \subfigure[Absolute Error]{
		\begin{minipage}[h]{0.45\textwidth}
		\centering
		\includegraphics[width=\textwidth]{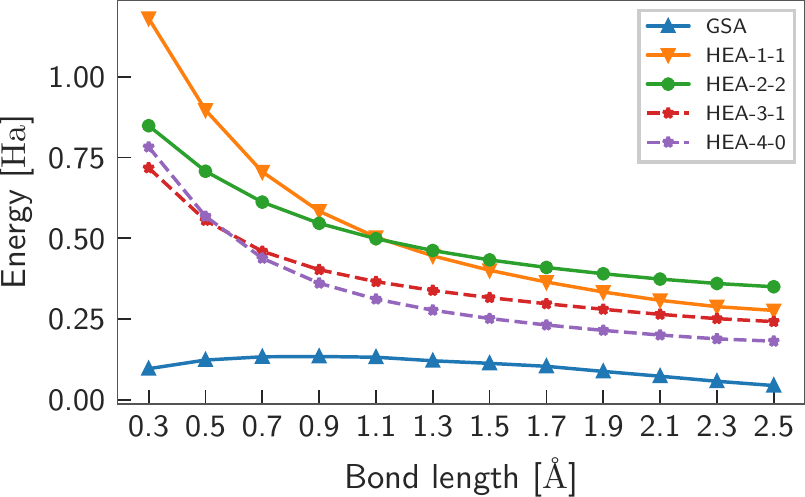}
		\end{minipage}
	}
  \caption{\label{fig:res_meta} The result of determining the energy profile of $\mathrm{H}_2$.}
\end{figure*}

\begin{table}[htp]
  \begin{center}
    \caption{The variance of absolute error and corresponding improvement.}
    \label{tab:meta_variance}
    \begin{tabular}{c| c c c c c}
    &\name& HEA-1-1& HEA-2-2& HEA-3-1& HEA-4-0\\
    \hline
    variance &0.001& 0.078& 0.024& 0.020& 0.033\\
    improvement & -&98.8\%& 96.1\%& 95.4\% & 97.2\%
    \end{tabular}
  \end{center}
  \end{table}

Simulations are conducted for $\mathrm{H}_2$, where $\Delta$ in $H(\Delta)$ represents the bond length. Training bonds are set as 0.5, 0.9, 1.3, 1.7, 2.1, 2.5 and 2.9. We set $N_l=4$ and other configurations as that for the VQE task of $\mathrm{H}_2$. Results are shown in Fig.~\ref{fig:res_meta} and Tab.~\ref{tab:meta_variance}. It can be concluded that our framework learns the energy profile with the conspicuously lower variance of absolute error, which indicates more information on chemical properties.

\section{Conclusion}\label{sec:conclusion}

In this paper, addressing the issues in mitigating the BP phenomenon in VQE, we propose a gradient-sensitive alternate framework for VQE with the variable ansatz strategy. We propose a theoretical framework that highlights the magnitude of gradient and exploits the alternate optimization scheme so that the local optimum can be avoided from the ansatz perspective. It can be theoretically proved that the result of the proposed theoretical framework is a subset of the results of the original VA-VQE. 

Then, based on the theoretical framework, a novel implementation named \name~is proposed with three stages. We reduce the size of the search space of the ansatz via applying gate commutation rules and establishing a bijection between the search space and the practical implementations of ansatzes to boost the time efficiency of the optimal ansatz and parameter determination. Exploiting the double $\epsilon$-greedy strategy based on the candidate tree, an initialization of parameters is determined, so that the training competitions are mitigated and the local optimal can be evaded from the parameter perspective. Based on the initialization of parameters, the \name~follows the framework with an elaborately modified genetic algorithm to find and train a quasi-optimal ansatz. 

Finally, we conduct numerical simulations on quantum chemistry to find the ground state energy of a quantum system. We adopt relatively fair criteria for measuring the performance of VA-VQE so that the transverse comparison among methods of VA-VQE can be clearly conducted. As a result, the \name~shows conspicuously better performance on average compared to the structure fixed HEA up to 87.9\% improvement in terms of absolute error. Furthermore, compared to the full-randomized RND, our framework obtains up to 36.0\% and 58.7\% improvement in terms of error and stability, respectively, with similar quantum costs. Moreover, we implement a modified version of our proposed framework for meta-VQE learning the energy profile of parameterized Hamiltonians to show the flexibility.

Although the proposed method obtained better results, further research is required to be conducted. For example, \name~requires a large number of hyperparameters to be adjusted such as $N_l$. Research on automatic hyperparameter adjustment may be meaningful for the VA-VQE.

\section*{Acknowledgements}
This work was supported by the National Science Foundation of China (No. 61871111 and No. 61960206005) and the Fundamental Research Funds for the Central Universities (No. 2242022k30006 and No. 2242022k30001).

\bibliographystyle{quantum}
\bibliography{citations}

\begin{thebibliography}{10}

\bibitem{grover1997Quantum}
Lov~K Grover.
\newblock ``Quantum computers can search arbitrarily large databases by a
  single query''.
\newblock \href{https://dx.doi.org/10.1103/physrevlett.79.4709}{Physical review
  letters {\bf 79}, 4709}~(1997).

\bibitem{shor1999Polynomialtime}
Peter~W Shor.
\newblock ``Polynomial-time algorithms for prime factorization and discrete
  logarithms on a quantum computer''.
\newblock In {{SIAM}} Review.
\newblock \href{https://dx.doi.org/10.1137/s0036144598347011}{Volume~41, pages
  303--332}.
\newblock {SIAM}~(1999).

\bibitem{harrow2009Quantum}
Aram~W. Harrow, Avinatan Hassidim, and Seth Lloyd.
\newblock ``Quantum algorithm for linear systems of equations''.
\newblock \href{https://dx.doi.org/10.1103/physrevlett.103.150502}{Physical
  review letters {\bf 103}, 150502}~(2009).

\bibitem{low2019Hamiltonian}
Guang~Hao Low and Isaac~L. Chuang.
\newblock ``Hamiltonian simulation by qubitization''.
\newblock \href{https://dx.doi.org/10.22331/q-2019-07-12-163}{Quantum {\bf 3},
  163}~(2019).

\bibitem{preskill2018Quantum}
John Preskill.
\newblock ``Quantum computing in the {{NISQ}} era and beyond''.
\newblock \href{https://dx.doi.org/10.22331/q-2018-08-06-79}{Quantum {\bf 2},
  79}~(2018).

\bibitem{peruzzo2014Variational}
Alberto Peruzzo, Jarrod McClean, Peter Shadbolt, Man-Hong Yung, Xiao-Qi Zhou,
  Peter~J. Love, Al{\'a}n {Aspuru-Guzik}, and Jeremy~L. O'brien.
\newblock ``A variational eigenvalue solver on a photonic quantum processor''.
\newblock \href{https://dx.doi.org/10.1038/ncomms5213}{Nature communications
  {\bf 5}, 1--7}~(2014).

\bibitem{bravo-prieto2020Variational}
Carlos {Bravo-Prieto}, Ryan LaRose, M.~Cerezo, Yigit Subasi, Lukasz Cincio, and
  Patrick~J. Coles.
\newblock ``Variational {{Quantum Linear Solver}}''~(2020).
\newblock  \href{http://arxiv.org/abs/1909.05820}{arXiv:1909.05820}.

\bibitem{yuan2019Theory}
Xiao Yuan, Suguru Endo, Qi~Zhao, Ying Li, and Simon~C. Benjamin.
\newblock ``Theory of variational quantum simulation''.
\newblock \href{https://dx.doi.org/10.22331/q-2019-10-07-191}{Quantum {\bf 3},
  191}~(2019).

\bibitem{wang2021Variational}
Xin Wang, Zhixin Song, and Youle Wang.
\newblock ``Variational quantum singular value decomposition''.
\newblock \href{https://dx.doi.org/10.22331/q-2021-06-29-483}{Quantum {\bf 5},
  483}~(2021).

\bibitem{meng2021Quantum}
Fan-Xu Meng, Ze-Tong Li, Xutao Yu, and Zaichen Zhang.
\newblock ``Quantum algorithm for {{MUSIC-based DOA}} estimation in hybrid
  {{MIMO}} systems''.
\newblock \href{https://dx.doi.org/10.1088/2058-9565/ac44dd}{Quantum Science
  and Technology}~(2021).

\bibitem{li2022Quantum}
Ze-Tong Li, Fan-Xu Meng, Xu-Tao Yu, and Zai-Chen Zhang.
\newblock ``Quantum algorithm for {{Laplacian}} eigenmap via {{Rayleigh}}
  quotient iteration''.
\newblock \href{https://dx.doi.org/10.1007/s11128-021-03347-y}{Quantum
  Information Processing {\bf 21}, 1--20}~(2022).

\bibitem{beer2020Training}
Kerstin Beer, Dmytro Bondarenko, Terry Farrelly, Tobias~J. Osborne, Robert
  Salzmann, Daniel Scheiermann, and Ramona Wolf.
\newblock ``Training deep quantum neural networks''.
\newblock \href{https://dx.doi.org/10.1038/s41467-020-14454-2}{Nature
  Communications {\bf 11}, 808}~(2020).

\bibitem{abbas2021Power}
Amira Abbas, David Sutter, Christa Zoufal, Aurelien Lucchi, Alessio Figalli,
  and Stefan Woerner.
\newblock ``The power of quantum neural networks''.
\newblock \href{https://dx.doi.org/10.1038/s43588-021-00084-1}{Nature
  Computational Science {\bf 1}, 403--409}~(2021).

\bibitem{biamonte2017Quantum}
Jacob Biamonte, Peter Wittek, Nicola Pancotti, Patrick Rebentrost, Nathan
  Wiebe, and Seth Lloyd.
\newblock ``Quantum machine learning''.
\newblock \href{https://dx.doi.org/10.1038/nature23474}{Nature {\bf 549},
  195--202}~(2017).

\bibitem{havlicek2019Supervised}
Vojt{\v e}ch Havl{\'i}{\v c}ek, Antonio~D. C{\'o}rcoles, Kristan Temme, Aram~W.
  Harrow, Abhinav Kandala, Jerry~M. Chow, and Jay~M. Gambetta.
\newblock ``Supervised learning with quantum-enhanced feature spaces''.
\newblock \href{https://dx.doi.org/10.1038/s41586-019-0980-2}{Nature {\bf 567},
  209--212}~(2019).

\bibitem{kandala2017Hardwareefficient}
Abhinav Kandala, Antonio Mezzacapo, Kristan Temme, Maika Takita, Markus Brink,
  Jerry~M Chow, and Jay~M Gambetta.
\newblock ``Hardware-efficient variational quantum eigensolver for small
  molecules and quantum magnets''.
\newblock \href{https://dx.doi.org/10.1038/nature23879}{Nature {\bf 549},
  242--246}~(2017).

\bibitem{du2020Quantum}
Yuxuan Du, Tao Huang, Shan You, Min-Hsiu Hsieh, and Dacheng Tao.
\newblock ``Quantum circuit architecture search: Error mitigation and
  trainability enhancement for variational quantum solvers''~(2020).
\newblock  \href{http://arxiv.org/abs/2010.10217}{arXiv:2010.10217}.

\bibitem{mcclean2018Barren}
Jarrod~R McClean, Sergio Boixo, Vadim~N Smelyanskiy, Ryan Babbush, and Hartmut
  Neven.
\newblock ``Barren plateaus in quantum neural network training landscapes''.
\newblock \href{https://dx.doi.org/10.1038/s41467-018-07090-4}{Nature
  communications {\bf 9}, 1--6}~(2018).

\bibitem{patti2021Entanglement}
Taylor~L. Patti, Khadijeh Najafi, Xun Gao, and Susanne~F. Yelin.
\newblock ``Entanglement devised barren plateau mitigation''.
\newblock \href{https://dx.doi.org/10.1103/physrevresearch.3.033090}{Physical
  Review Research {\bf 3}, 033090}~(2021).

\bibitem{ortizmarrero2021EntanglementInduced}
Carlos Ortiz~Marrero, M{\'a}ria Kieferov{\'a}, and Nathan Wiebe.
\newblock ``Entanglement-{{Induced Barren Plateaus}}''.
\newblock \href{https://dx.doi.org/10.1103/PRXQuantum.2.040316}{PRX Quantum
  {\bf 2}, 040316}~(2021).

\bibitem{cerezo2021Cost}
Marco Cerezo, Akira Sone, Tyler Volkoff, Lukasz Cincio, and Patrick~J. Coles.
\newblock ``Cost function dependent barren plateaus in shallow parametrized
  quantum circuits''.
\newblock \href{https://dx.doi.org/10.1038/s41467-021-21728-w}{Nature
  communications {\bf 12}, 1--12}~(2021).

\bibitem{sack2022Avoiding}
Stefan~H. Sack, Raimel~A. Medina, Alexios~A. Michailidis, Richard Kueng, and
  Maksym Serbyn.
\newblock ``Avoiding barren plateaus using classical shadows''~(2022).
\newblock  \href{http://arxiv.org/abs/2201.08194}{arXiv:2201.08194}.

\bibitem{holmes2022Connecting}
Zo{\"e} Holmes, Kunal Sharma, M.~Cerezo, and Patrick~J. Coles.
\newblock ``Connecting ansatz expressibility to gradient magnitudes and barren
  plateaus''.
\newblock \href{https://dx.doi.org/10.1103/PRXQuantum.3.010313}{PRX Quantum
  {\bf 3}, 010313}~(2022).
\newblock  \href{http://arxiv.org/abs/2101.02138}{arXiv:2101.02138}.

\bibitem{wang2021Noiseinduced}
Samson Wang, Enrico Fontana, Marco Cerezo, Kunal Sharma, Akira Sone, Lukasz
  Cincio, and Patrick~J. Coles.
\newblock ``Noise-induced barren plateaus in variational quantum algorithms''.
\newblock \href{https://dx.doi.org/10.1038/s41467-021-27045-6}{Nature
  communications {\bf 12}, 1--11}~(2021).

\bibitem{rattew2020Domainagnostic}
Arthur~G. Rattew, Shaohan Hu, Marco Pistoia, Richard Chen, and Steve Wood.
\newblock ``A {{Domain-agnostic}}, {{Noise-resistant}}, {{Hardware-efficient
  Evolutionary Variational Quantum Eigensolver}}''~(2020).
\newblock  \href{http://arxiv.org/abs/1910.09694}{arXiv:1910.09694}.

\bibitem{zhang2021Differentiable}
Shi-Xin Zhang, Chang-Yu Hsieh, Shengyu Zhang, and Hong Yao.
\newblock ``Differentiable {{Quantum Architecture Search}}''~(2021).
\newblock  \href{http://arxiv.org/abs/2010.08561}{arXiv:2010.08561}.

\bibitem{ostaszewski2021Reinforcement}
Mateusz Ostaszewski, Lea~M. Trenkwalder, Wojciech Masarczyk, Eleanor Scerri,
  and Vedran Dunjko.
\newblock ``Reinforcement learning for optimization of variational quantum
  circuit architectures''.
\newblock In Advances in Neural Information Processing Systems.
\newblock \href{https://dx.doi.org/10.48550/arXiv.2103.16089}{Volume~34, pages
  18182--18194}.
\newblock {Curran Associates, Inc.}~(2021).

\bibitem{bilkis2022Semiagnostic}
M.~Bilkis, M.~Cerezo, Guillaume Verdon, Patrick~J. Coles, and Lukasz Cincio.
\newblock ``A semi-agnostic ansatz with variable structure for quantum machine
  learning''~(2022).
\newblock  \href{http://arxiv.org/abs/2103.06712}{arXiv:2103.06712}.

\bibitem{zhang2021Mutual}
Zi-Jian Zhang, Thi~Ha Kyaw, Jakob Kottmann, Matthias Degroote, and Alan
  {Aspuru-Guzik}.
\newblock ``Mutual information-assisted adaptive variational quantum
  eigensolver''.
\newblock \href{https://dx.doi.org/10.1088/2058-9565/abdca4}{Quantum Science
  and Technology}~(2021).

\bibitem{meng2021Quantuma}
Fanxu Meng, Ze-Tong Li, Xu-Tao Yu, and Zai-Chen Zhang.
\newblock ``Quantum {{Circuit Architecture Optimization}} for {{Variational
  Quantum Eigensolver}} via {{Monto-Carlo Tree Search}}''.
\newblock \href{https://dx.doi.org/10.1109/tqe.2021.3119010}{IEEE Transactions
  on Quantum Engineering}~(2021).

\bibitem{chivilikhin2020MoGVQE}
D.~Chivilikhin, A.~Samarin, V.~Ulyantsev, I.~Iorsh, A.~R. Oganov, and
  O.~Kyriienko.
\newblock ``{{MoG-VQE}}: {{Multiobjective}} genetic variational quantum
  eigensolver''~(2020).
\newblock  \href{http://arxiv.org/abs/2007.04424}{arXiv:2007.04424}.

\bibitem{elsken2019Neural}
Thomas Elsken, Jan~Hendrik Metzen, and Frank Hutter.
\newblock ``Neural {{Architecture Search}}: {{A Survey}}''~(2019).
\newblock  \href{http://arxiv.org/abs/1808.05377}{arXiv:1808.05377}.

\bibitem{mitarai2018Quantum}
Kosuke Mitarai, Makoto Negoro, Masahiro Kitagawa, and Keisuke Fujii.
\newblock ``Quantum circuit learning''.
\newblock \href{https://dx.doi.org/10.1103/physreva.98.032309}{Physical Review
  A {\bf 98}, 032309}~(2018).

\bibitem{wolfe1969Convergence}
Philip Wolfe.
\newblock ``Convergence conditions for ascent methods''.
\newblock \href{https://dx.doi.org/10.1137/1011036}{SIAM review {\bf 11},
  226--235}~(1969).

\bibitem{sim2019Expressibility}
Sukin Sim, Peter~D. Johnson, and Al{\'a}n Aspuru-Guzik.
\newblock ``Expressibility and entangling capability of parameterized quantum
  circuits for hybrid quantum-classical algorithms''.
\newblock \href{https://dx.doi.org/10.1002/qute.201900070}{Advanced Quantum
  Technologies {\bf 2}, 1900070}~(2019).

\bibitem{nakaji2021Expressibility}
Kouhei Nakaji and Naoki Yamamoto.
\newblock ``Expressibility of the alternating layered ansatz for quantum
  computation''.
\newblock \href{https://dx.doi.org/10.22331/q-2021-04-19-434}{Quantum {\bf 5},
  434}~(2021).

\bibitem{cerezo2021Higher}
Marco Cerezo and Patrick~J. Coles.
\newblock ``Higher order derivatives of quantum neural networks with barren
  plateaus''.
\newblock \href{https://dx.doi.org/10.1088/2058-9565/abf51a}{Quantum Science
  and Technology {\bf 6}, 035006}~(2021).

\bibitem{arrasmith2021Effect}
Andrew Arrasmith, M.~Cerezo, Piotr Czarnik, Lukasz Cincio, and Patrick~J.
  Coles.
\newblock ``Effect of barren plateaus on gradient-free optimization''.
\newblock \href{https://dx.doi.org/10.22331/q-2021-10-05-558}{Quantum {\bf 5},
  558}~(2021).

\bibitem{stilckfranca2021Limitations}
Daniel Stilck~Fran{\c c}a and Raul {Garcia-Patron}.
\newblock ``Limitations of optimization algorithms on noisy quantum devices''.
\newblock \href{https://dx.doi.org/10.1038/s41567-021-01356-3}{Nature Physics
  {\bf 17}, 1221--1227}~(2021).

\bibitem{pesah2021Absence}
Arthur Pesah, M.~Cerezo, Samson Wang, Tyler Volkoff, Andrew~T. Sornborger, and
  Patrick~J. Coles.
\newblock ``Absence of barren plateaus in quantum convolutional neural
  networks''.
\newblock \href{https://dx.doi.org/10.1103/physrevx.11.041011}{Physical Review
  X {\bf 11}, 041011}~(2021).

\bibitem{volkoff2021Large}
Tyler Volkoff and Patrick~J. Coles.
\newblock ``Large gradients via correlation in random parameterized quantum
  circuits''.
\newblock \href{https://dx.doi.org/10.1088/2058-9565/abd891}{Quantum Science
  and Technology {\bf 6}, 025008}~(2021).

\bibitem{skolik2021Layerwise}
Andrea Skolik, Jarrod~R. McClean, Masoud Mohseni, Patrick {van der Smagt}, and
  Martin Leib.
\newblock ``Layerwise learning for quantum neural networks''.
\newblock \href{https://dx.doi.org/10.1007/s42484-020-00036-4}{Quantum Machine
  Intelligence {\bf 3}, 1--11}~(2021).

\bibitem{grant2019Initialization}
Edward Grant, Leonard Wossnig, Mateusz Ostaszewski, and Marcello Benedetti.
\newblock ``An initialization strategy for addressing barren plateaus in
  parametrized quantum circuits''.
\newblock \href{https://dx.doi.org/10.22331/q-2019-12-09-214}{Quantum {\bf 3},
  214}~(2019).

\bibitem{verdon2019Learning}
Guillaume Verdon, Michael Broughton, Jarrod~R. McClean, Kevin~J. Sung, Ryan
  Babbush, Zhang Jiang, Hartmut Neven, and Masoud Mohseni.
\newblock ``Learning to learn with quantum neural networks via classical neural
  networks''~(2019).
\newblock  \href{http://arxiv.org/abs/1907.05415}{arXiv:1907.05415}.

\bibitem{you2020Greedynas}
Shan You, Tao Huang, Mingmin Yang, Fei Wang, Chen Qian, and Changshui Zhang.
\newblock ``Greedynas: {{Towards}} fast one-shot nas with greedy supernet''.
\newblock In Proceedings of the {{IEEE}}/{{CVF Conference}} on {{Computer
  Vision}} and {{Pattern Recognition}}.
\newblock \href{https://dx.doi.org/10.1109/CVPR42600.2020.00207}{Pages
  1999--2008}.
\newblock {IEEE}~(2020).

\bibitem{cervera-lierta2021Metavariational}
Alba {Cervera-Lierta}, Jakob~S. Kottmann, and Al{\'a}n {Aspuru-Guzik}.
\newblock ``Meta-variational quantum eigensolver: {{Learning}} energy profiles
  of parameterized hamiltonians for quantum simulation''.
\newblock \href{https://dx.doi.org/10.1103/PRXQuantum.2.020329}{PRX Quantum
  {\bf 2}, 020329}~(2021).

\bibitem{bergholm2020Pennylane}
Ville Bergholm, Josh Izaac, Maria Schuld, Christian Gogolin, M.~Sohaib Alam,
  Shahnawaz Ahmed, Juan~Miguel Arrazola, Carsten Blank, Alain Delgado, Soran
  Jahangiri, Keri McKiernan, Johannes~Jakob Meyer, Zeyue Niu, Antal Sz{\'a}va,
  and Nathan Killoran.
\newblock ``{{PennyLane}}: {{Automatic}} differentiation of hybrid
  quantum-classical computations''~(2020).
\newblock  \href{http://arxiv.org/abs/1811.04968}{arXiv:1811.04968}.

\end{thebibliography}






\onecolumn\newpage
\appendix

\section{Examples}\label{appendix:example}

\begin{example}[A VQE Task]\label{eg:vqe}
  The VQE is to find the ground state energy of a quantum system where the dynamic can be described by a Hamiltonian $H$. Setting the observable $O=H$, the function $f(x)=x$, the input state $\rho = \left|0\right>\left<0\right|$ ,and the ansatz $U(\bm{\theta})$ as the hardware efficient ansatz (c.f. Appendix \ref{appendix:hea}), where $\bm{\theta} \in \left(-\pi,\pi\right]^{d}$ and $d$ is the number of parameters, the problem is 
  \begin{equation}
   \min_{\bm{\theta}\in \left(-\pi,\pi\right]^{d}} C(\bm{\theta}) = {\rm Tr}\left[H U(\bm{\theta}) \left|0\right>\left<0\right| U^\dagger(\bm{\theta})\right].
  \end{equation}
\end{example}

\begin{example}[Not Bijection]\label{eg:not_bijection}
  For a $2$-qubit quantum system, $N_l = 2$, define $\astate_1$ and $\astate_2$ as $(I^{q_1},R_y^{q_2})$ and $(R_y^{q_1}, I^{q_2})$, respectively. Then, the practical implementations of two ansatzes described by paths
  \begin{align*}
    \apath_1 &= (\anode^0, \anode_1^1:=(\anode^0,\astate_1),\anode_1^2:=(\anode^1,\astate_2)),\\
    \apath_2 &= (\anode^0, \anode_2^1:=(\anode^0,\astate_2),\anode_2^2:=(\anode^1,\astate_1))
  \end{align*}
  are identical.
\end{example}

\begin{example}[Mergeable \& Deletable]\label{eg:mnd}
  For a $4$-qubit quantum system, $N_l = 2$, define 
  \begin{align*}
    \astate_1&:=(I^{q_1},R_y^{q_2}, \mathrm{CNOT}^{q_3,q_4}),\\
    \astate_2&:=(R_z^{q_1}, R_y^{q_2},\mathrm{CNOT}^{q_3,q_4}),\\
    \astate_3&:=(R_z^{q_1}, R_y^{q_2}).
  \end{align*}
  Then, the practical implementations of two ansatzes described by paths
  \begin{align*}
    \apath_1 &= (\anode^0, \anode_1^1:=(\anode^0,\astate_1),\anode_1^2:=(\anode^1,\astate_2)),\\
    \apath_2 &= (\anode^0, \anode_2^1:=(\anode^0,\astate_3),\anode_2^2:=(\anode^1,(~))),
  \end{align*}
  are equivalent.
\end{example}

\begin{example}[Size of Search Space]\label{eg:space_size}
  The number of possible paths in a $4$-qubit quantum system under the setting as described in Sec.~\ref{subsec:hardware_constraints} with and without the gates commutation rules $N_{w/}$ and $N_{w/o}$ with respect to $N_l$ are demonstrated in Tab.\ref{tab:space_size}
  \begin{table}[htp]
    \begin{center}
      \caption{The number of possible paths with respect to $N_l$.}
      \label{tab:space_size}
      \begin{tabular}{c c c c}
      $N_l$& $N_{w/}$& $N_{w/o}$& $\frac{N_{w/}}{N_{w/o}}$\\
      \hline
      1 & 56& 567& 9.88\%\\
      2 & 11768& 321489& 3.66\%\\
      3 & 2859977& 182284263& 1.57\%\\
      \end{tabular}
    \end{center}
    \end{table}
\end{example}

\section{Hardware efficient ansatz} \label{appendix:hea}
The hardware efficient ansatz (HEA) was first proposed in \cite{kandala2017Hardwareefficient} to conduct VQE for small molecules and quantum magnets. It is constructed in a layer-by-layer fashion in which each layer has an identical structure. Specifically, a layer is organized by native single-qubit gates on all qubits and naturally available entangling interactions. In our setting, based on the native gates as described in Sec.~\ref{subsec:hardware_constraints}, we exploit the HEA with layers as depicted in Fig.~\ref{fig:hea}.

\begin{figure}[ht]
  \centering
	\includegraphics[width=0.45\textwidth]{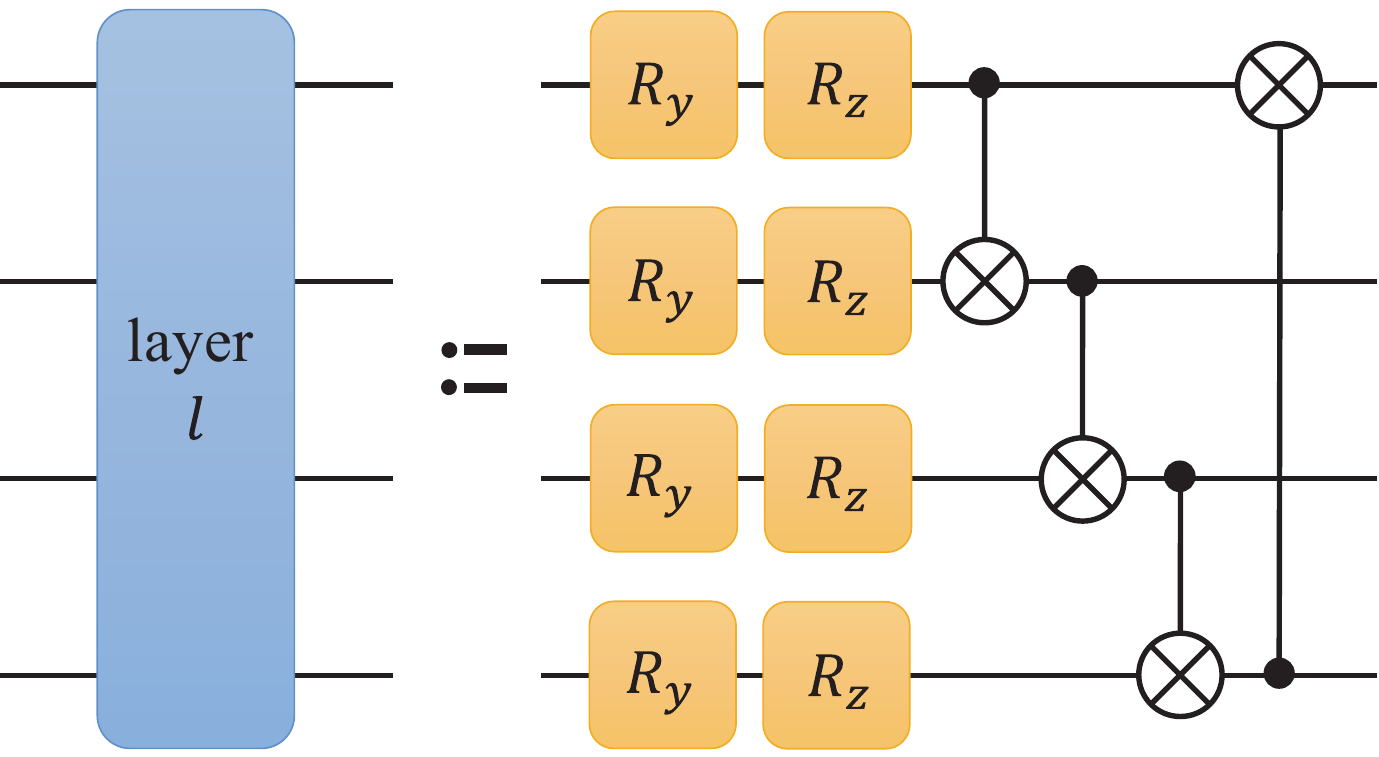}
  \caption{\label{fig:hea}{Layer construction used in HEA.}}
\end{figure}

\section{Proof of Theorem \ref{th:convergence}}\label{appendix:proof_convergence}

Solution of Prob.~\ref{prob:vqe} is a set of parameters 
\begin{equation}
  \Theta^*=\{\bm{\theta}^*|C(\bm{\theta}^*) = \min_{\bm{\theta}\in \mathbb{D}^d}C(\bm{\theta})\}.
\end{equation}

\begin{equation}
  \left\{(U^*,\Theta^*)\left|
  \begin{gathered}C(U^*,\bm{\theta}^*) = \min_{U,\bm{\theta}}C(U,\bm{\theta}),\\ \forall \bm{\theta}^*\in\Theta^*
  \end{gathered}
  \right.\right\}.
\end{equation}

Given an arbitrary valid $\mathbb{P}$, we divide $\mathbb{S}$ into four disjoint sets
\begin{equation}
  \mathbb{U}_+^*\equiv\left\{U\left|C(U,\bm{\theta}_U) = \min_{U,\bm{\theta}}C(U,\bm{\theta})\right.\right\},
\end{equation}
\begin{equation}
  \mathbb{U}_+\equiv\left\{U\left|
    \exists \bm{\theta}^\prime,C(U,\bm{\theta}^\prime) = \min_{U,\bm{\theta}}C(U,\bm{\theta}),~
    C(U,\bm{\theta}_U) > \min_{\bm{\theta}}C(U,\bm{\theta})
    \right.\right\},
\end{equation}
\begin{equation}
  \mathbb{U}_-^*\equiv\left\{U\left|
    \forall \bm{\theta}^\prime,C(U,\bm{\theta}^\prime) > \min_{U,\bm{\theta}}C(U,\bm{\theta}),~
    C(U,\bm{\theta}_U) = \min_{\bm{\theta}}C(U,\bm{\theta})
    \right.\right\},
\end{equation}
\begin{equation}
  \mathbb{U}_-\equiv\left\{U\left|
    \forall \bm{\theta}^\prime,C(U,\bm{\theta}^\prime) > \min_{U,\bm{\theta}}C(U,\bm{\theta}),~
    C(U,\bm{\theta}_U) > \min_{\bm{\theta}}C(U,\bm{\theta})
    \right.\right\}.
\end{equation}
Note that $\mathbb{U}_+^* \cup \mathbb{U}_+\cup\mathbb{U}_-^*\cup\mathbb{U}_- = \mathbb{S}$. We prove the theorem by proving 
\begin{align}
  \lim_{t\to\infty}\mathbb{U}_n \cap \left(\mathbb{U}_-^*\cup\mathbb{U}_-\right)=\emptyset,\label{eq:cond_1}\\
  \lim_{t\to\infty}\mathbb{U}_n =\lim_{t\to\infty} \mathbb{U}_+^*,\label{eq:cond_2}
\end{align}
where $t$ indicates the iteration of the alternate optimization loop.

\begin{lemma}\label{lemma:finally_trained}
  $U\in\lim_{t\to\infty}\mathbb{U}_+^*\cup\mathbb{U}_-^*$, $\forall U\in \mathbb{U}_+\cup\mathbb{U}_-$ such that $\frac{\left\|\nabla C(U,\bm{\theta}_U)\right\|_2}{|\bm{\theta}_U|}>0$.
\end{lemma}
\begin{proof}
  By solving the Prob.~\ref{prob:multi_obj}, all ansatzes $U\in\mathbb{U}_+\cup\mathbb{U}_-$ such that 
  $$
  \frac{\left\|\nabla C(U,\bm{\theta}_U)\right\|_2}{|\bm{\theta}_U|}=\max_{\bm{\theta}}\frac{\left\|\nabla C(U,\bm{\theta})\right\|_2}{|\bm{\theta}|} >0
  $$
  are included in $\mathbb{U}_n$. For any $U\in\mathbb{U}_n$, $\bm{\theta}_U\in\mathbb{P}$ is updated by $\bm{\theta}^*$ such that $C(U,\bm{\theta}^*) = \min_{\bm{\theta}}C(U,\bm{\theta})$, which means any $U\in \mathbb{U}_n$ is an element of $U \in \mathbb{U}_+^*\cup\mathbb{U}_-^*$ in the next iteration. Therefore, any $U\in \mathbb{U}_+\cup\mathbb{U}_-$ such that $\frac{\left\|\nabla C(U,\bm{\theta}_U)\right\|_2}{|\bm{\theta}_U|}>0$ will be selected as an element in $\mathbb{U}_n$ if $t\to\infty$. Then, we have $U\in\lim_{t\to\infty}\mathbb{U}_+^*\cup\mathbb{U}_-^*$, $\forall U\in \mathbb{U}_+\cup\mathbb{U}_-$ such that $\frac{\left\|\nabla C(U,\bm{\theta}_U)\right\|_2}{|\bm{\theta}_U|}>0$.
\end{proof}

Further, it is straightforward that $U\in\lim_{t\to\infty}\mathbb{U}_+^*$, $\forall U\in \mathbb{U}_+$ such that $\frac{\left\|\nabla C(U,\bm{\theta}_U)\right\|_2}{|\bm{\theta}_U|}>0$ and $V\in\lim_{t\to\infty}\mathbb{U}_-^*$,  $\forall V\in \mathbb{U}_-$ such that $\frac{\left\|\nabla C(V,\bm{\theta}_V)\right\|_2}{|\bm{\theta}_V|}>0$. Since $C(U,\bm{\theta}_U) = \min_{\bm{\theta}}C(U,\bm{\theta})$ for all $U \in \mathbb{U}_+^*\cup\mathbb{U}_-^*$, we have that
\begin{equation}
  \frac{\left\|\nabla C(U,\bm{\theta}_U)\right\|_2}{|\bm{\theta}_U|}=0, ~\forall U\in\mathbb{U}_+^*\cup\mathbb{U}_-^*.
\end{equation}
Therefore, if $\mathbb{U}_+^*$ is not empty, then inequalities Eq.~(\ref{eq:cost_ineq}) and Eq.~(\ref{eq:grad_ineq}) are simultaneously established for all $U\in\mathbb{U}_-^*$, $V\in\mathbb{U}_+^*$, which means $\mathbb{U}_n\cap \mathbb{U}_-^* = \emptyset$. Then, we have that 
\begin{equation}
  \lim_{t\to\infty} \mathbb{U}_n\cap (\mathbb{U}_-^*\cup\mathbb{{U}_-}) = \emptyset,
\end{equation}
if $\lim_{t\to\infty} U_+^*\ne\emptyset$.

The simultaneous establishment of Eq.~(\ref{eq:exist_optimal_theta}) and Eq.~(\ref{eq:not_local_maximum}) of $V$ 
implies that $V\in\mathbb{U}_+$ with $\frac{\left\|\nabla C(V,\bm{\theta}_V)\right\|_2}{|\bm{\theta}_V|}>0$ or $V\in\mathbb{U}_+^*$, which indicates
\begin{equation}
  \lim_{t\to\infty}U_+^*\ne\emptyset.
\end{equation}

It is obvious that $\mathbb{U}_+^* \subset \mathbb{U}_n$ and $\mathbb{U}_n \subset \mathbb{U}_+^*\cup\mathbb{U}_+$. Since 
\begin{equation}
  \lim_{t\to\infty}U_+ \equiv \left\{U\left|\begin{gathered}\exists \bm{\theta}^\prime,C(U,\bm{\theta}^\prime) = \min_{U,\bm{\theta}}C(U,\bm{\theta}),\\
  \forall \delta >0,~C(U,\bm{\theta}_U\in \lim_{t\to\infty}\mathbb{P}) = \max_{\bm{\theta}\in \Theta_{U,\delta}}C(U,\bm{\theta})
  \end{gathered}
  \right.\right\},
\end{equation}
where $\Theta_{U,\delta} = \{\bm{\theta}\left|\|\bm{\theta}-\bm{\theta}_U\|_2<\delta\right.\}$, we have that $\lim_{t\to\infty}(\mathbb{U}_n\cup\mathbb{U}_+)=\emptyset$. Therefore, Eq.~(\ref{eq:cond_1}) and Eq.~(\ref{eq:cond_2}) are established.

\section{Hyperparameters}\label{appendix:hyperparameters}

\begin{table}[h]
  \begin{center}
    \caption{Hyperparameters.}
    \label{tab:hyperparameters}
    \begin{tabular}{c | c c}
      scope & name& notation\\
      \hline
      \hline
      \multirow{4}{*}{global} & maximum number of layers & $N_l$\\
      & reference step size& $\alpha_0$\\
      & convergence threshold& $\xi$\\
      & probabilities for double $\epsilon$-greedy strategy& $\epsilon_1$, $\epsilon_2$\\
      \hline
      \multirow{5}{*}{pool training}& number of sampled ansatzes& $N_{s1}$\\
      & maximum number of ranks for updating parameters& $N_{r1}$\\
      & stable threshold for terminating& $N_{t1}$\\
      & maximum iteration times in the prethermalization & $N_{i0}$\\
      & maximum iteration times in the main process &$N_{i1}$\\
      \hline
      \multirow{5}{*}{alternate training}& population size& $N_{s2}$\\
      & maximum number of ranks for updating parameters& $N_{r2}$\\
      & stable threshold for terminating& $N_{t2}$\\
      & optimization step per generation& $N_{o}$\\
      &maximum iteration times& $N_{i2}$\\
      \hline
      VQE retraining& maximum iteration times& $N_{i3}$
    \end{tabular}
  \end{center}
\end{table}

\section{Genetic Operators}\label{appendix:genetic_operator}
In this appendix, we provide genetic operators the \name~applied in detail. We only consider the asexual genetic operators for simplicity. Consequently, only one ansatz is generated from each survival ansatz. All genetic operators comply with the same structure. First, a new practical implementation of ansatz is randomly generated. Then, the implementation is refined by the gate commutation rule. Since the bijection between the search space and the practical implementations of ansatzes is established, the path representation of the new implementation is finally output. We introduce the genetic operators identified in the first step.

\paragraph{Mutation:} The mutation operator randomly selects a node $\anode$ in the survival ansatz. Then, a child node $\anode_{new}$ of the parent node of $\anode$ is randomly sampled. Therefore, a new implementation of ansatz substituting $\anode$ by $\anode_{new}$ is generated.

\paragraph{Deletion:} The deletion operator randomly deletes a node in the survival ansatz and connects its parent node and child node. Then, a new implementation is generated.

\paragraph{Amplification:} The amplification operator randomly selects a node $\anode$ in the survival ansatz. Then, a child node of $\anode$ is inserted into the ansatz, and therefore a new implementation is generated.

\section{Pseudocodes for Algorithms} \label{appendix:alg}
In this appendix, we list the pseudocodes for algorithms. Note that the task $\mathbb{T}$ is generally omitted as the input parameters of algorithms without confusion.

\begin{algorithm}[ht]
  \caption{\label{al:alt_opt}alternate VA-VQE}
  \KwIn{$\mathbb{S}$, $\mathbb{P}_0$}
  \SetKwFunction{PO}{ParameterOptimization}
  \Begin{
    $\mathbb{P}\leftarrow \mathbb{P}_0$;\tcp*[f]{Initialize $\mathbb{P}$}\\ 
    $\mathbb{U}_n \leftarrow \left\{\,\right\}$;\tcp*[f]{Initialize $\mathbb{U}_n$}\\
    \While{$\mathbb{U}_n$ and $\mathbb{P}$ are not converged}{
      $\mathbb{U}_n \leftarrow \mathbf{SolveProb\ref{prob:multi_obj}}(\mathbb{S},\mathbb{P})$;\\
      \ForEach{$U \in \mathbb{U}_n$}{
        $\bm{\theta}_U\in\mathbb{P}\leftarrow\mathbf{SolveProb\ref{prob:vqe}}(U)$;
      }
    }
  }
  \Return{$\mathbb{U}_n$, $\mathbb{P}$.}
\end{algorithm}

\begin{algorithm}[ht]
  \caption{\label{al:gsa}\name}
  \KwIn{$\mathbb{G}$}
  \SetKwFunction{PO}{ParameterOptimization}
  \Begin{
    $\mathbb{S}(N_l)\leftarrow$ Construct search space by $\mathbb{G}$;\\
    $\mathbb{P}_0,T \leftarrow\mathbf{Pool Training}(\mathbb{S}(N_l))$;\\
    $U^*, \bm{\theta}_{U^*} \leftarrow\mathbf{AlternateTraining}(\mathbb{S}(N_l),\mathbb{P}_0,T)$;\\
    $U^*, \bm{\theta}^* \leftarrow\mathbf{VQE Training}(U^*,\bm{\theta}_{U^*})$;
  }
  \Return{$U^*$, $\bm{\theta}^*$.}
\end{algorithm}

\begin{algorithm}[ht]
  \caption{\label{al:sample}USampling}
  \KwIn{$\mathbb{S}(N_l)$, $T$, $\epsilon_1$, $\epsilon_2$}
  \Begin{
    $J_1 \sim \left\{(True, \epsilon_1),(False, 1-\epsilon_1)\right\};$\\
    \eIf{$J_1$ is $True$}{
      $\anode\leftarrow\anode^0\in T$;\\
      $U\leftarrow (\anode)$;\\
      $\eta \sim \left\{(1, \epsilon_2),(0, 1-\epsilon_2)\right\}$;\\
      \For{$\anode$ is not a leaf node}{
        $\mathbb{V}_c \leftarrow \mathbf{GetChildNodes(\anode)}$;\\
        $\anode^\prime\sim \left\{(\anode^\prime_k\in\mathbb{V}_c, {\rm Pr}(\anode^\prime_k;\eta))\right\}$;\\
        $U \leftarrow U \oplus (\anode^\prime)$;\\
        $\anode\leftarrow\anode^\prime$;
      }
    }{
      $U \sim \mathbb{S}(N_l)$;\\
    }
  }
  \Return{$U$.}
\end{algorithm}

\begin{algorithm}[ht]
  \caption{\label{al:main_pool_training}PoolTrainingMainProcess}
  \KwIn{$\mathbb{S}(N_l)$, $P$, $T$}
  \SetKwFunction{PO}{ParameterOptimization}
  \Begin{
    $\mathbb{S}\leftarrow\mathbb{S}(N_l)$;\\
    $t\leftarrow 0$;\\
    $c_l^\prime \leftarrow 1;$\\
    $i\leftarrow 1$;\\
    \For{$i\le N_{i1}$}{
      $\mathbb{U}_s\leftarrow \left\{\right\}$;\\
      \For{$\left|\mathbb{U}_s\right| < N_{s1}$}{
        $U\leftarrow \mathbf{USampling}(\mathbb{S},T,\epsilon_1,\epsilon_2)$;\\
        $\mathbb{U}_s\leftarrow \mathbb{U}_s \cup \left\{U\right\}$;\\
      }
      $\mathbb{U}_n(N_{r1})\leftarrow\mathbf{SolveProb\ref{prob:multi_obj}}(\mathbb{U}_s,P)$;\\
      \ForEach{$U\in\mathbb{U}_n(N_{r1})$}{
        $\bm{\theta}\leftarrow\mathbf{GetParameters}(U,P)$;\\
        $\alpha \leftarrow \mathbf{GetStepSize}(\alpha_0)$;\\
        $\bm{\theta}\leftarrow \bm{\theta} - \alpha \nabla C(U,\bm{\theta})$;\\
        $P\leftarrow \mathbf{UpdatePool}(\bm{\theta},P)$;\\
        $T \leftarrow$ Append $U$ on the candidate tree $T$;\\
      }
      \eIf{$c_l(\anode^0) = c_l^\prime$}{
        $t\leftarrow t+1$;\\
        \If{$t = N_{t1}$}{
          $\mathbf{break}$;\\
        }
      }{
        $t\leftarrow 0$;\\
        $c_l^\prime \leftarrow c_l(\anode^0)$;\\
      }
      $i\leftarrow i+1$;\\
    }
  }
  \Return{$P$, $T$.}
\end{algorithm}

\begin{algorithm}[ht]
  \caption{\label{al:pre_pool_training}PoolTrainingPrethermalization}
  \KwIn{$\mathbb{S}(N_l)$, $P$, $T$}
  \Begin{
    $\mathbb{S}\leftarrow\mathbb{S}(N_l)$;\\
    $i\leftarrow 1$\\
    $\epsilon_1^\prime = {(i-1)\epsilon_1}/{N_{i0}}$;\\
    \For{$i\le N_{i1}$}{
      $\mathbb{U}_s\leftarrow \left\{\right\}$;\\
      \For{$\left|\mathbb{U}_s\right| < N_{s1}$}{
        $U\leftarrow \mathbf{USampling}(\mathbb{S},T,\epsilon_1^\prime,\epsilon_2)$;\\
        $\mathbb{U}_s\leftarrow \mathbb{U}_s \cup \left\{U\right\}$;\\
        $j\leftarrow j+1$;\\
      }
      $\mathbb{U}_n(N_{r1})\leftarrow\mathbf{SolveProb\ref{prob:multi_obj}}(\mathbb{U}_s,P)$;\\
      \ForEach{$U\in\mathbb{U}_n(N_{r1})$}{
        $\bm{\theta}\leftarrow\mathbf{GetParameters}(U,P)$;\\
        $\alpha \leftarrow \mathbf{GetStepSize}(\alpha_0)$;\\
        $\bm{\theta}\leftarrow \bm{\theta} - \alpha \nabla C(U,\bm{\theta})$;\\
        $P\leftarrow \mathbf{UpdatePool}(\bm{\theta},P)$;\\
        $T \leftarrow$ Append $U$ on the candidate tree $T$;\\
      }
      $i\leftarrow i+1$;\\
    }
    $\mathbb{P}_0\leftarrow \mathbf{Expand}(P)$;\\
  }
  \Return{$P$, $T$.}
\end{algorithm}

\begin{algorithm}[ht]
  \caption{\label{al:pool_training}PoolTraining}
  \KwIn{$\mathbb{S}(N_l)$}
  \Begin{
    $P\leftarrow$ Construct parameter pool with $\bm{0}$ initialization;\\
    $T\leftarrow \{\anode^0\}$;\\
    $P,T \leftarrow \mathbf{PoolTrainingPrethermalization}(\mathbb{S}(N_l),P,T)$;\\
    $P,T \leftarrow \mathbf{PoolTrainingMainProcess}(\mathbb{S}(N_l),P,T)$;\\
    $\mathbb{P}_0\leftarrow\mathbf{Expand}(P)$;
  }
  \Return{$\mathbb{P}_0$, $T$.}
\end{algorithm}

\begin{algorithm}[ht]
  \caption{\label{al:alt_training}AlternateTraining}
  \KwIn{$\mathbb{S}(N_l)$, $\mathbb{P}_0$, $T$}
  \Begin{
    $\mathbb{P}\leftarrow\mathbb{P}_0$\\
    $\mathbb{U}_s\leftarrow \left\{\right\}$;\\
    $U_{best}\leftarrow \mathbf{None}$;\\
    \For{$\left|\mathbb{U}_s\right|\le N_{s2}$}{
      $U\leftarrow \mathbf{USampling}(\mathbb{S},T,\epsilon_1,\epsilon_2)$;\\
      $\mathbb{U}_s\leftarrow \mathbb{U}_s \cup \left\{U\right\}$;\\
    }
    $i\leftarrow 1$;\\
    \For{$i\le N_{i2}$}{
      $\mathbb{U}_n(N_{r2}),\mathbb{U}_n(N_{r}^\prime),\mathbb{U}_n(N_{r}^\prime + 1)\leftarrow\mathbf{SolveProb\ref{prob:multi_obj}}(\mathbb{U}_s,\mathbb{P})$;\\
      \ForEach{$U\in\mathbb{U}_n(N_{r2})$}{
        $j\leftarrow 1$;\\
        \For{$j\le N_{o}$}{
          $\alpha \leftarrow \mathbf{GetStepSize}(\alpha_0)$;\\
          $\bm{\theta}_U\leftarrow \bm{\theta}_U - \alpha \nabla C(U,\bm{\theta}_U)$;\\
          $j\leftarrow j+1$;\\
        }
      }
      $\mathbb{U}_{survival}\leftarrow\mathbf{GetSurvivals}(\mathbb{U}_n(N_r^\prime),\mathbb{U}_n(N_r^\prime+1))$;\\
      \If{$\exists U\in \mathbb{U}_{survival}$ such that $C(U,\bm{\theta}_U) < C(U_{best},\bm{\theta}_{U_{best}})$}{
        $U_{best}\leftarrow \mathbf{None}$;\\
      }
      $\mathbb{U}_{new}\leftarrow\{\}$;\\
      \ForEach{$U\in\mathbb{U}_{survival}$}{
        $U_{new}\leftarrow \mathbf{RandomGeneticOperator}(U)$;\\
        $\mathbb{U}_{new}\leftarrow\mathbb{U}_{new} \cup \{U_{new}\}$;\\
        $\alpha \leftarrow \mathbf{GetStepSize}(\alpha_0)$;\\
        \If{$\alpha\frac{\left\|\nabla C(U,\bm{\theta}_U)\right\|_2}{|\bm{\theta}_U|}<\xi$}{
          $\mathbb{U}_{survival} \leftarrow \mathbb{U}_{survival} \setminus \{U\}$;\\
          \If{$C(U,\bm{\theta}_U) < C(U_{best},\bm{\theta}_{U_{best}})$}{
            $U_{best}\leftarrow U$;\\
          }
        }
      }
      $U_s\leftarrow \mathbb{U}_{survival}\cap \mathbb{U}_{new}$;\\
      \For{$\left|\mathbb{U}_s\right|\le N_{s2}$}{
        $U\leftarrow \mathbf{USampling}(\mathbb{S},T,\epsilon_1,\epsilon_2)$;\\
        $\mathbb{U}_s\leftarrow \mathbb{U}_s \cup \left\{U\right\}$;\\
      }
      \If{$U_{best}$ has been preserved for $N_{t2}$ generations}{
          $\mathbf{break}$;\\
      }
    }
    $U^*\leftarrow \arg\min_{U\in\mathbb{U}_s\cup\{U_{best}\}}C(U,\bm{\theta}_U)$;\\
  }
  \Return{$U^*$, $\bm{\theta}_{U^*}$.}
\end{algorithm}

\begin{algorithm}[ht]
  \caption{\label{al:vqe_retraining}VQERetraining}
  \KwIn{$U$, $\bm{\theta}_0$}
  \Begin{
    $\bm{\theta}\leftarrow\bm{\theta}_0$;\\
    $i\leftarrow 1$\\
    \For{$i\le N_{i3}$}{
      $\alpha \leftarrow \mathbf{GetStepSize}(\alpha_0)$;\\
      $\bm{\theta}\leftarrow \bm{\theta} - \alpha \nabla C(U,\bm{\theta})$;\\
      $i\leftarrow i+1$;\\
      \If{$\alpha\frac{\left\|\nabla C(U,\bm{\theta})\right\|_2}{|\bm{\theta}|}<\xi$}{
      $\mathbf{break}$;\\
      }
    }
    $U^*\leftarrow U, \bm{\theta}^*\leftarrow\bm{\theta}$;
  }
  \Return{$U^*$, $\bm{\theta}^*$.}
\end{algorithm}

\end{document}